\PassOptionsToPackage{table,dvipsnames}{xcolor}
\documentclass[sigconf, screen]{acmart}

\usepackage[ruled,vlined,linesnumbered,noend]{algorithm2e}
\usepackage{xspace}
\usepackage{bbm}
\usepackage{arydshln}
\usepackage{subcaption}
\usepackage[nameinlink]{cleveref}

\usepackage{enumitem}

\AtBeginDocument{%
  }

\setcopyright{acmlicensed}
\copyrightyear{2018}
\acmYear{2018}
\acmDOI{XXXXXXX.XXXXXXX}

\acmConference[Conference acronym 'XX]{Make sure to enter the correct
  conference title from your rights confirmation emai}{June 03--05,
  2018}{Woodstock, NY}
\acmISBN{978-1-4503-XXXX-X/18/06}

\usepackage[breakable, skins]{tcolorbox}
\tcolorboxenvironment{example}{
  colback=gray!10!white,
  boxrule=0pt,
  boxsep=0pt,
  left=3pt,right=3pt,top=3pt,bottom=3pt,
  oversize=0pt,
  sharp corners,
  breakable,
  frame empty,
  opacityframe=0,
  enhanced,
}

\definecolor{ClusterGreen}{HTML}{8fbbda}
\definecolor{ClusterBlue}{HTML}{96d096}
\definecolor{FilledCircle}{HTML}{333333}

\SetKwComment{Comment}{$\blacktriangleright$~}{}

\SetCommentSty{commentfont}

\newcommand{\eg}{\emph{e.g.,}\xspace}
\newcommand{\ie}{\emph{i.e.,}\xspace}

\newcommand*\first{\emph{(i)}\xspace}
\newcommand*\second{\emph{(ii)}\xspace}

\newcommand{\unstable}[3]{#1\text{-unstable~for~}\allowbreak#2(\allowbreak#3)\allowbreak}

\newcommand{\stability}{\textsf{Stability}}

\newcommand{\kStableZoneBound}{k\text{-}\textsf{SB}_{f(D)}^t}
\newcommand{\akStableZoneBound}{\alpha\text{-}k\text{-}\textsf{SB}_{f(D)}^t}
\newcommand{\algoName}{\textsf{LStability}\xspace}

\newcommand{\resChanges}{\textsf{RC}}

\newcommand{\outerZone}{\textsf{E}}
\newcommand{\vol}{\textsf{Vol}}

\newcommand{\denseRegionAlgo}{\texttt{Detect-\allowbreak Dense-\allowbreak Region}\xspace}

\newtcolorbox{revaBox}{
  colback=BrickRed!10!white,
  boxrule=0pt,
  boxsep=1pt,
  left=6pt,right=6pt,top=6pt,bottom=6pt,
  oversize=2pt,
  sharp corners,
  breakable,
  frame empty,
  opacityframe=0,
  enhanced,
  borderline west = {2pt}{0pt}{BrickRed}
}

\newtcolorbox{revbBox}{
  colback=NavyBlue!10!white,
  boxrule=0pt,
  boxsep=1pt,
  left=6pt,right=6pt,top=6pt,bottom=6pt,
  oversize=2pt,
  sharp corners,
  breakable,
  frame empty,
  opacityframe=0,
  enhanced,
  borderline west = {2pt}{0pt}{NavyBlue}
}

\newtcolorbox{revcBox}{
  colback=Green!80!black!10!white,
  boxrule=0pt,
  boxsep=1pt,
  left=6pt,right=6pt,top=6pt,bottom=6pt,
  oversize=2pt,
  sharp corners,
  breakable,
  frame empty,
  opacityframe=0,
  enhanced,
  borderline west = {2pt}{0pt}{Green!80!black}
}

\newif\iftechreport

\theoremstyle{acmdefinition}

\Crefname{question}{Question}{Questions}

\theoremstyle{acmplain}

\Crefname{remark}{Remark}{Remarks}

\newenvironment{sketch}{%
  \renewcommand{\proofname}{Proof Sketch}\proof}{\endproof}

\begin{document}
\title{Local Stability of Rankings}

\author{Felix S. Campbell}
\orcid{0000-0003-3888-1491}
\affiliation{%
  \institution{Ben-Gurion University of the Negev}
  \city{Beersheba}
  \state{}
  \country{Israel}
}
\email{felixsal@post.bgu.ac.il}

\author{Yuval Moskovitch}
\orcid{0000-0002-0325-7392}
\affiliation{%
  \institution{Ben-Gurion University of the Negev}
  \city{Beersheba}
  \state{}
  \country{Israel}
}
\email{yuvalmos@bgu.ac.il}

\renewcommand{\shortauthors}{Felix S. Campbell \& Yuval Moskovitch}

\begin{abstract}
Rankings play a crucial role in decision-making. However, if minor changes to items significantly alter their rankings, the quality of the decisions being made can be compromised. The stability of ranking is a measure used to assess how modifications to the ranking algorithm or data affect results. While previous work has focused on stability of the ranking under changes to the algorithm, we introduce a novel measure we refer to as \emph{local stability}. Local stability indicates the effect of minor changes to the values of an item in the ranking on its rank.
Our proposed definition furthermore takes into account the presence of multiple items with similar qualities in the ranking, called \emph{dense regions}, permitting minor modifications to swap the positions of items within the region. We show that computing this measure in general is hard, and in turn propose a relaxation of the definition to admit approximation.

We present \first \algoName, a sampling-based algorithm for approximating local stability, on which we make probably-approximately-correct-type guarantees through the use of concentration inequalities, and \second \denseRegionAlgo, an algorithm based on this approach to detect the dense region an item lies in, if it exists. We introduce a number of optimizations to our algorithms to improve their scalability and efficiency. We validate our proposed framework through an extensive suite of experiments, including case studies highlighting the utility of our definitions.
\end{abstract}




\received{20 February 2007}
\received[revised]{12 March 2009}
\received[accepted]{5 June 2009}


\maketitle

\section{Introduction}
\label{sec:introduction}
Ranking items or individuals by their quantitative attributes often plays a key role in a wide range of domain applications, \eg in academia~\cite{nrc-ranking,goto-rankings-considered-helpful}, hiring \cite{LinkedIn}, or e-commerce~\cite{ecommerce-ranking}. A major underlying assumption of ranking is that a higher ranking reflects a meaningful improvement in utility over lower ranked items. However, if minor modifications to the data result in significant shifts in an item's position in the ranking, this fundamental assumption is undermined. 
As a simple example, consider the results of a motorsport event: if the driver finishing 1\textsuperscript{st} did so with a margin of only $0.2$ seconds, the distinction from the runner-up may be practically negligible, whereas a $20$-second lead would indicate a clearly superior performance. While the improvement in utility is clear for these simple cases, defining margins between items in more complex settings, such as learning-to-rank (LtR) \cite{fairness-ranking-survey-pt1} or rankings based on scores computed by multiple attributes, is not straightforward, as we next illustrate. 

\begin{table}[t]
    \caption{Top-10 computer science departments ranked according to adjusted average publication count. Each group of rows highlighted by a single color represents a dense region in the ranking in which the departments are similar.}
    \scalebox{0.83}{
        \begin{tabular}{c:crr:r}
        \hline
        & \textbf{University} & \textbf{AI Pubs.} & \textbf{Systems Pubs.} & \textbf{Score $\downarrow$} \\
        \hline
        \rowcolor{CornflowerBlue!20} $t_{1}$ & Lakefront University & 44 & 36 & 39.2  \\
        \rowcolor{CornflowerBlue!20} $t_{2}$ & Dempster University & 42 & 35 & 37.9  \\
        \rowcolor{CornflowerBlue!20} $t_{3}$ & Western Polytechnic & 43 & 33 & 36.7  \\
        \rowcolor{Peach!20} $t_{4}$ & Prairie University & 23 & 25 & 25.4  \\
        \rowcolor{Peach!20} $t_{5}$ & Ogden University & 22 & 24 & 24.4  \\
        \rowcolor{Peach!20} $t_{6}$ & Kedzie Institute & 20 & 24 & 23.8  \\
        \rowcolor{Peach!20} $t_{7}$ & University of Blue Island & 21 & 22 & 22.7  \\
        \rowcolor{ForestGreen!20} $t_{8}$ & Plainfield College & 7 & 13 & 11.9  \\
        \rowcolor{ForestGreen!20}$t_{9}$ & Irving University & 8 & 11 & 11.0  \\
        \rowcolor{ForestGreen!20} $t_{10}$ & Kimball College & 6 & 10 & 9.6 \\
        \bottomrule
        \end{tabular}
    }
    \label{tab:universities}
\end{table}
\begin{example}
    \label{ex:running}
    When choosing a university to enroll in, students often consult a ranking of potential universities compiled by third parties (\eg CSRankings \cite{CSRankings}) in order to pick a program that can best help them achieve their goals. As a running example, we consider a ranking of fictitious universities based on a simplified version of CSRankings for the sake of exposition.  \Cref{tab:universities} shows the top-$10$ universities in descending order of their adjusted average publication counts, computed by 
    $$\varphi(t) = \sqrt[17]{(t.\texttt{AI Pubs.} + 1)^5 \cdot (t.\texttt{Systems Pubs.} + 1)^{12}}$$
    Intuitively, this score is the geometric mean of adjusted publication counts for the $5$ and $12$ subfields of AI and systems, respectively.

    In this case, while Lakefront University is ranked~1\textsuperscript{st} according to the ranking function, with just four fewer publications in systems (keeping the rest of the data unchanged), it would be~3\textsuperscript{rd}. As the number of publications can slightly vary over time, this may raise questions such as \emph{``how much does this university deserve its position in the ranking?''} or \emph{``was this university a close second in the ranking?''} and questions on a broader level, such as \emph{``how likely are the two universities to shift in the ranking?''}
\end{example}

Quantifying the well-foundedness of the positions of items in a ranking is referred to as ranking \emph{stability},  
 and has been studied in previous works~\cite{AJMS18,sensitivity-vectors,stability-multigroup-fairness,ranking-nutrition,mithraranking,fairly-evaluating-and-scoring}. In particular, in \cite{AJMS18}, ranking stability is used to measure the robustness of the ranking to changes in the ranking function used to produce it.
This definition gives a \emph{global} ranking stability score and provides insight into how stable the ranking is with respect to uncertainty in the ranking function. 

Rankings frequently feature items with comparable qualities, 
 such as universities of comparable standing, which can create \emph{dense regions} within the ranking, where small changes can reasonably lead to position swaps among the items. This phenomenon was recognized, for example, in the National Research Council's ranking of doctoral programs \cite{nrc-ranking}, which provides ranges of possible rankings for each university to account for uncertainty in the collected data.

\begin{example}
    \label{ex:dense-regions}
    The universities in \Cref{tab:universities} are divided into three dense regions. Each dense region corresponds to a group of universities with similar scores generated by the given ranking function. As shown in \Cref{ex:running}, with small (hypothetical) modifications to the number of publications ($4$ in total), the~1\textsuperscript{st} ranked university becomes ~3\textsuperscript{rd}. In contrast, the change in the number of publications that would shift Lakefront University to be ranked 4\textsuperscript{th} is $17$. The presence of dense regions in a ranking may impact decision-making. For example, consider a student admitted to the universities ranked 4\textsuperscript{th} through 7\textsuperscript{th}. Although the top choice should be the highest-ranked option, since all of these universities belong to a dense region, they may be considered as roughly equivalent in quality. Consequently, other factors, such as location or program fit, might be prioritized over the precise ranking order.
\end{example}

The ranking stability measure presented in~\cite{AJMS18} has a coarse-grained nature, treating all changes in the ranking equally---for instance, a transposition of a single pair of items is as significant as a complete reversal of the ranking. Thus, this measure may overlook dense regions in the data.  
This motivates us to propose the concept of \emph{local stability}, which considers stability as a property of individual items within the ranking rather than the ranking as a whole. Considering stability as a property of an individual item leads to a more natural treatment of dense regions, which are themselves properties of local areas within the ranking. Our contribution may be summarized as follows.

\paragraph{\textbf{Local stability}} 
We define a local stability measure aiming to quantify the effect of modifications to data on the outcome of the ranking for a given tuple by characterizing the magnitudes of changes (or \emph{refinements}) to the tuple that can change its position significantly. To account for the effects of modifications to the data on the ranking we leverage the notion of counterfactual reasoning \cite{counterfactual-expl-survey,GeCo,DiCE,FACET}. Intuitively, we consider small hypothetical (counterfactual) modifications to a single tuple in the data to determine its margin relative to other tuples in the ranking. Restricting attention to a single tuple yields a local stability measure that reflects a best-case analysis, assuming no uncertainty in the remaining tuples of the database, which would otherwise exacerbate instability in many cases.
Intuitively, lower magnitude modifications that can change the tuple's position significantly correspond to a lower local stability.
Our definition hinges on a user-given definition of what a \emph{significant} change in positions is, and is specified by choosing a value $k$ representing a range of positions around the tuple's original ranking. This parameter $k$ can then be used to account for dense regions when evaluating the local stability.  
Furthermore, our definition takes into account a \emph{set of reasonable changes} which may be determined by domain-specific knowledge or in a data-driven manner, and bounds the magnitude of the refinements under consideration.
Since local stability is focused on changes in the data, we are able to treat the ranking function as a black-box, allowing our definition to accommodate any ranking function, including those derived by complex LtR models.

In a nutshell, refinements may be partially ordered by the magnitude of modifications applied to each attribute value. For instance, continuing  ~\Cref{ex:running}, changing the number of AI publications by 3 and systems publications by 2 is considered to be a larger modification compared to only altering the number of AI publications by 2. Utilizing this partial order, we can define a \emph{stable zone}, where, intuitively, any perturbations to the values of a given tuple do not cause its position in the ranking to change by more than the given value of $k$. The local stability value is then defined as the relative portion of the stable zone, compared to the overall space of reasonable changes given by the user. 

While useful and informative, computing the local stability of a tuple is intractable. 
To circumvent this challenge, we propose a relaxed definition, allowing the stable zone to possibly contain a small number of refinements that can modify the tuple's position in the ranking by more than $k$ (which we refer to as $k$-\emph{unstable} refinements), however keeping the probability that such a refinement is sampled (uniformly) from the stable zone low. 

\paragraph{\textbf{Estimating local stability}}
We present a two-stage sampling-based algorithm that first samples from the space of reasonable changes to construct (an approximate) stable zone. It then verifies that the generated stable zone is approximately stable by sampling refinements within it. We show that by taking enough samples, verifying whether the computed stable zone is indeed approximately stable (\ie includes a low number of unstable refinements) can be done with high probability.
We propose three optimizations to improve the performance of our sampling-based estimator. 

\paragraph*{\textbf{Detecting dense regions}}
The local stability of a tuple evaluates how well it aligns with its assigned position in the ranking, and may be used to answer questions such as \emph{``how much does this university deserve its position in the ranking?''} as demonstrated in \Cref{ex:running}. 
Alternatively, an important related problem is to determine the potential bounds within which a tuple's position could vary in the ranking. Towards this end, we propose \denseRegionAlgo, whose goal is to identify, given a tuple, a value $k$ that covers the extent of the dense region of the tuple. We do so by leveraging our local stability definition to understand where the margin between tuples (in counterfactual scenarios) is small enough to imply the existence and extent of a dense region.

\paragraph{\textbf{Experimental evaluation}} 
Finally, we validate our definition and proposed algorithms, \algoName and \denseRegionAlgo, through an extensive experimental evaluation using both real and synthetic data, showing the usefulness of the approach and the efficiency of the algorithms.
We present two case studies in which our definition leads to insights for rankings on real data. For example, we show that for the ranking of NBA players, the learned ranking function overfits Joel Embiid, implying his high ranking is ill-founded. Furthermore, we find that the CSRankings\cite{CSRankings} ranking is fairly locally stable for most of the top-10 universities, supporting its reliability. Our experimental results also indicate that \denseRegionAlgo is able to accurately identify dense regions observed in the ranking. We conclude the experimental evaluation with a comparison of our local stability measure to the global stability measure as defined in~\cite{AJMS18}, demonstrating how the two definitions can lead to divergent interpretations of the ranking.

\paragraph*{Paper organization}
The rest of the paper is organized as follows.
 In \Cref{sec:stability}, we provide the background and formalizations necessary and then define exact and approximate local stability for rankings. In \Cref{sec:computing}, we develop techniques to approximate local stability of an item in a ranking, and propose optimizations to our algorithm in \Cref{sec:optimizations}. We present a heuristic to detect the dense region around an item in \Cref{sec:detecting-dense-regions}. In \Cref{sec:experiments}, we provide case studies and empirical analysis on the performance of our approach and a survey of related work is given in \Cref{sec:related-work}. Finally, we draw conclusions and discuss future directions in \Cref{sec:conclusion}.

\section{Problem Formulation}

\begin{table}[t] 
    \centering
    \caption{Summary of notations}
    \label{tab:notations}

    \begin{tabular}{@{} l p{0.7\columnwidth} @{}} 
        \toprule
        \textbf{Notation} & \textbf{Description} \\ 
        \midrule
        
        $\varepsilon(t)$                            & Tuple obtained by refining $t$ by $\varepsilon$ \\
        $D_{t \rightarrow t'}$                      & Database $D$ where tuple $t$ is replaced by $t'$ \\ 
        $\Delta_{f(D)}(t, t')$                      & Change in positions for $t$ when replaced by $t'$ under $f(D_{t \rightarrow t'})$ \\
        $\preceq$                                   & Refinement containment relation \\
        $\outerZone(\cdot)$                         & Set of non-containing refinements \\
        $\resChanges$                               & Set of reasonable changes \\
        $\kStableZoneBound$                         & $k$-stable zone boundary \\
        $\akStableZoneBound$                        & $\alpha$-$k$-stable zone boundary \\
        \bottomrule
    \end{tabular}
\end{table}

\label{sec:stability}
Local stability measures the effect of small modifications to a tuple on its position in the ranking while also taking into account dense regions, where multiple tuples have a similar quality and small changes can reasonably lead to position swaps among the tuples. We assume that data modifications can be made independently, in the sense that changing one tuple does not require modifying any other tuple to preserve the semantics of the data. Under this assumption, we focus on modifications to individual tuples.
We consider a \emph{ranking} to be a permutation of a set of tuples~\cite{sharp}. A \emph{ranking function} is a function $f$ mapping a database $D$ to a permutation of its tuples. We use $f(D)$ to denote the ranking resulting when applying $f$ on a database $D$, and $f(D)[t]$ for the position of $t$ in the ranking.
When working with any database $D$, we consider only its numerical attributes. This assumption is standard and follows prior work (\eg~\cite{whynotyet,AJMS18,sharp}).

\subsection{Refinements}
We aim to understand how changes in the values of different attributes of tuples affect their positions in the ranking. To do this, we first define \emph{refinements} of a tuple, which describe possible modifications to the tuple and their impact on its values.

\begin{definition}[Refinement]
    \label{def:refinement}
     Given a database $D$ with numeric attributes $a_1,\dots,\allowbreak a_n$, a \emph{refinement} is a vector $\varepsilon = \langle \varepsilon_{1}, \dots, \varepsilon_n \rangle \in \mathbb{R}^n$. 
     Given a refinement $\varepsilon$, the refined tuple $\varepsilon(t)$ is the tuple $t'$ such that for all $i \in [n]$,
     $$t'.a_{i} = t.a_{i} + \varepsilon_i$$
\end{definition}

\begin{example}
    \label{ex:refinement}
    Let us continue from \Cref{ex:running}. Considering a refinement $\varepsilon = \langle -10, 5\rangle$, we have that $\varepsilon(t_1)$ is the tuple $(34,41)$.
\end{example}

To understand how changes in the attributes of tuples affect their positions in the ranking, we define the notion of \emph{position change} between two tuples.

\begin{definition}[Position change]
    \label{def:pos-change}
    For a database $D$ and a tuple $t$, let $D_{t\rightarrow t'}$ be the database obtained by replacing $t$ with $t'$ $$D_{t\rightarrow t'} = D \setminus \{t\} \cup \{t'\}$$ 
    The change of positions attained by substituting $t'$ for $t$ is 
    $$\Delta_{f(D)}(t, t') = |f(D)[t] - f(D_{t \rightarrow t'})[t']|$$ 
\end{definition}

\begin{example}
    \label{ex:position-change}
    Consider the refinement $\varepsilon = \langle -10, -5 \rangle$. The refined tuple $\varepsilon(t_1) = (34, 31)$ has a score of $32.9$, which is less than the score of $t_3$. Therefore, we have $\Delta_{f(D)}(t_1, \varepsilon(t_1)) = 2$ since $\varepsilon(t_1)$ ranks 3\textsuperscript{rd} in $D_{t \rightarrow \varepsilon(t_1)}$.   
\end{example}

Recall that our goal is to determine the effect of small modifications to a tuple on its position in the ranking, while also taking into account dense regions. To this end, we introduce a parameter $k$, which defines a range of positions around the tuple with which we would like to evaluate its local stability\footnote{In this paper, we discuss the symmetric case where the range extends in both directions of the ranking equally, for clarity of presentation. Extending our definitions and methods to make these ranges asymmetric is straightforward.}. This parameter can then be used to represent the size of the area around a tuple (\ie tuples that are ranked at most $k$ positions above or below) with approximately equal quality, \ie the dense region. We assume $k$ is given as an input, although we propose a heuristic in \Cref{sec:detecting-dense-regions} for determining a value of $k$ suitable for capturing the dense region around the tuple, if it exists.
We are now ready to define a $k$-(un)stable refinement.

\begin{definition}[$k$-(un)stable refinement]
    \label{def:unstable-refinement}
    Given a database $D$, a ranking function $f$, a refinement $\varepsilon$, a value 
    $k$, and a tuple $t\in D$, we say $\varepsilon(t)$ is \emph{$k$-stable} for $t$ over $f(D)$ if and only if $\Delta_{f(D)}(t, \varepsilon(t)) \leq k$. Otherwise, we say that $\varepsilon(t)$ is \emph{$k$-unstable}. 
\end{definition}

\begin{example}
    Consider again the refinement $\varepsilon = \langle -10, -5 \rangle$, as in \Cref{ex:position-change}. Since $\Delta_{f(D)}(t_1, \varepsilon(t_1)) = 2$, we have that $\varepsilon(t_1)$ is $1$-unstable for $t_1$ over $f(D)$.   
\end{example}

\subsection{Local Stability}
\label{sec:def-local-stability}

In order to quantify the effect of small modifications to a tuple has on its position in the ranking, we use the magnitude of the changes required to do so.
As smaller changes are more likely than larger changes, we are primarily interested in the \emph{smallest} modifications that cause an item to move more than $k$ positions in the ranking. For instance, in our running example, having $\pm$2 systems publications is more likely than having $\pm10$ systems publications. With this in mind, we first establish a partial order over refinements.

\begin{definition}[Refinement containment]
    \label{def:containment}
We say that a refinement $\varepsilon$ is \emph{contained} in a refinement $\varepsilon'$ (and $\varepsilon'$ \emph{contains} $\varepsilon$), and denote it by $\varepsilon \preceq \varepsilon'$, if for every $i \in [n]$, $|\varepsilon_i| \leq |\varepsilon'_i|$. 
\end{definition}

Refinement containment is defined with respect to absolute values as we are primarily concerned with the \emph{magnitudes} of the refinements being made. For example, the refinement $\varepsilon' = \langle 10, 6\rangle$ contains the refinement $\varepsilon = \langle 10, -5\rangle$ since $\left|-5\right| \leq \left|6\right|$. Naturally, we write $\varepsilon \prec \varepsilon'$ when $\varepsilon \preceq \varepsilon'$ and there is an $i$ such that $|\varepsilon_i| < |\varepsilon'_i|$, \ie it is the \emph{strict} refinement containment relation. 
We can now define a $k$-stable zone boundary of a tuple $t$ as follows. 

\begin{definition}[Stable zone boundary]
    \label{def:stable-zone-boundary}
    Given a database $D$, a tuple $t \in D$, a ranking function $f$ and a value $k$, let $U$ be the set of refinements such that $\varepsilon(t)$ is $k$-unstable, and $U^+$ be the projection of $U$ onto $\mathbb{R}^n_{\geq 0}$. Then, the \emph{$k$-stable zone boundary} of $t$ is the set $$\kStableZoneBound =Min_\prec(U^+)$$ where $$Min_\prec(R) = \{\varepsilon \in R \mid \not\exists \varepsilon'\in R,~\varepsilon' \prec \varepsilon \}$$%
\end{definition}

Intuitively, the stable zone boundary consists of the $k$-unstable refinements of \emph{minimal} magnitude, \ie there is no $k$-unstable refinement containing any refinement of the boundary (this is also known as a \emph{skyline} \cite{skyline}).

\begin{example}
\label{ex:stable-zone}
Consider a dataset $D$ with two attributes and a ranking function $f$ over tuples in $D$. \Cref{fig:space-of-refinements} illustrates the space of refinements around a tuple $t\in D$ where each axis correspond to the values of $t$ in each one of the attributes. Refinements in the green region are $k$-stable, and the red area consists of $k$-(un)stable refinements (\ie refinements in the set $U$). \Cref{fig:k-stable-zound-boundary} shows the projection of the refinements onto $\mathbb{R}^n_{\geq 0}$. In this illustration, all the refinements in the green area remain $k$-stable, whereas the lighter, cross-hatched red region corresponds to refinement magnitudes admitting both $k$-stable and $k$-unstable refinements in the projected space. Consequently, the stable zone boundary $\kStableZoneBound$ is the curved black line that separates the green zone from the red zone.
\end{example}

\begin{definition}[Stable zone]\label{def:stable-zone}
    For a set of refinements $R$, we denote by $\outerZone(R)$ the set of refinements that do not contain any refinement in $R$, \ie $$\outerZone(R) = \{\varepsilon \in \mathbb{R}^n \mid \not\exists \varepsilon' \in R,~ \varepsilon' \preceq \varepsilon\}$$
    We refer to $\outerZone(\kStableZoneBound)$ as the \emph{stable zone}.
\end{definition}
By \Cref{def:stable-zone-boundary}, any refinement in the stable zone is guaranteed to be $k$-stable. Continuing \Cref{ex:stable-zone}, the green region in \Cref{fig:k-stable-zound-boundary} is the stable zone $\outerZone(\kStableZoneBound)$.

The local stability of a tuple is defined with respect to a set of reasonable changes specified by the user, which bounds the values of possible refinements to small changes to the data for which the user believes should not result in significant changes to the tuple's ranking (either leveraging domain-specific knowledge or in a data-driven manner as we do in \Cref{sec:case-studies}). More formally, we define $\varepsilon^{max}$ as a vector $\langle \varepsilon_{1}^{max}, \dots, \varepsilon_n^{max} \rangle \in \mathbb{R}^n_{\geq 0}$, and consider only refinements $\varepsilon$ such that $\varepsilon \preceq \varepsilon^{max}$.
\begin{example}
\label{ex:reasonable-changes}
 In the context of \Cref{ex:running}, the set of reasonable changes may depend on the amount of activity in the research area. For example, if we assume AI is a more active area than systems at these universities, we may assume that there is a larger variance in the number of AI papers than systems papers. Therefore, we may choose, \eg $\varepsilon^{max} = \langle 5, 3 \rangle$, encoding the assumption that each university might have $\pm 5$ AI publications, and $\pm 3$ systems publications.
\end{example}

\begin{figure}[t]
    \begin{subfigure}[t]{0.48\linewidth}
        \centering
        \includegraphics[width=1.66cm]{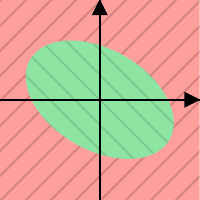}
        \caption{(Un)stable refinements}
        \label{fig:space-of-refinements}
    \end{subfigure}
    \begin{subfigure}[t]{0.48\linewidth}
        \centering
        \includegraphics[width=1.7cm]{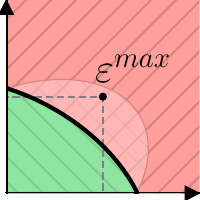}
        \caption{$\kStableZoneBound$ \& $\resChanges$}
        \label{fig:k-stable-zound-boundary}
    \end{subfigure}
    \caption{Illustrated examples of definitions in \Cref{sec:stability}}
    \label{fig:definitions}
\end{figure}

We let $\resChanges$ denote the set of refinements contained within the vector $\varepsilon^{max}$ (\ie all $\varepsilon \preceq \varepsilon^{max}$). For example, in \Cref{fig:definitions}, $\resChanges$ is the set of refinements whose projection into $\mathbb{R}^n_{\geq 0}$ is in the area delineated by the dashed gray lines. We are finally ready to define the local stability of a tuple. 

\begin{definition}[Local stability]
    \label{def:localstability}
    The \emph{local stability} of a tuple $t$ with respect to a database $D$, a ranking function $f$, a set of reasonable changes $\resChanges$, and a parameter $k$ is defined as
    $$\stability_{f(D)}(t, k \mid \resChanges) = \frac{\vol\big(\resChanges \cap \outerZone(\kStableZoneBound) \big)}{\vol(\resChanges)}$$
    where $\vol(\cdot)$ is the volume of the sets (\ie its Lebesgue measure).
\end{definition}

Intuitively, local stability measures the relative size between the stable zone restricted to the set of reasonable changes and the set of reasonable changes. 

\begin{example}
\label{ex:rc}
For the $\varepsilon^{max}$ and $\resChanges$ shown in \Cref{fig:k-stable-zound-boundary},
the local stability is the ratio between the green area restricted to $\resChanges$ and the area beneath $\resChanges$. In this case, the local stability would be high, since the area is mostly green.
\end{example}

Computing $\kStableZoneBound$ is a non-trivial task. In fact, we can show
\begin{theorem}
    \label{thm:hardness_bound_comp}
    Unless $\textsf{FP} = \textsf{\#P}$, there is no polynomial-time algorithm for computing $\kStableZoneBound$ given a database $D$, a tuple $t \in D$, a ranking function $f$, and a value $k$.
\end{theorem}
\begin{sketch}
    We prove this by showing that if $\kStableZoneBound$ can be computed in polynomial time, then computing $|\kStableZoneBound|$ is possible in polynomial time as well. However, we show that \textsf{\#DNF}, a well-known \textsf{\#P-complete} problem \cite{dnf-sharp-p-complete}, is polynomial-time reducible to computing $|\kStableZoneBound|$. Therefore, if there is such an algorithm to compute $\kStableZoneBound$, then $\textsf{FP} = \textsf{\#P}$. We relieve the details to \Cref{sec:apdx-hardness}.
\end{sketch}

Moreover, we expect that even if the stable zone boundary $\kStableZoneBound$ is given, computing the stability measure remains intractable since it is closely related to the problem of computing the \emph{hypervolume indicator}, which computes the volume of a union of hyperrectangles anchored at the origin and is known to be \textsf{\#P-hard}~\cite{hypervolume-survey,hypervolume-fpras}.

\subsection{$\alpha$-Local Stability}

Following the complexity of local-stability computation, we define an approximate version of the problem. We start by defining an $\alpha$-$k$-stable zone boundary as a set of refinements that is a $k$-stable boundary up to some error. Namely, the probability of sampling a $k$-unstable refinement from the set of refinements under an $\alpha$-$k$-stable zone boundary is at most $\alpha$ when samples are drawn uniformly.

\begin{definition}[$\alpha$-$k$-stable zone boundary]
    \label{def:alpha-stable-zone-boundary}
    Given a database $D$, a tuple $t \in D$, a ranking function $f$, a value $k$, and a set of reasonable changes $\resChanges$, let $U$ be the set of refinements such that $\varepsilon(t)$ is $k$-unstable, and $U^+$ be the projection of $U$ onto $\mathbb{R}^n_{\geq 0}$. Then, an \emph{$\alpha$-$k$-stable zone boundary} of $t$ is a set $\akStableZoneBound \subseteq U^+$ such that
    \begin{enumerate}
        \item [\first] $\akStableZoneBound = Min_\prec(\akStableZoneBound)$
        \item [\second] $\mathbb{P}_{\varepsilon \sim \mathcal{U}(\resChanges)}[\varepsilon(t) \text{ is } k\text{-unstable} \mid \varepsilon \in \outerZone(\akStableZoneBound)] \leq \alpha$
    \end{enumerate}
    where $\mathcal{U}$ is the uniform distribution.
\end{definition}

When $\alpha$ is small, we can assure that with high probability, refinements in $RC \cap \outerZone(\akStableZoneBound)$ when made to the tuple will not cause its position in the ranking to shift by more than $k$. We can now define $\alpha$-local stability.

\begin{definition}[$\alpha$-local stability]
    \label{def:alpha-localstability}
    An \emph{$\alpha$-local stability} of a tuple $t$ with respect to a database $D$, a ranking function $f$, an $\akStableZoneBound$ $Sb$, a set of reasonable changes \resChanges, and a parameter $k$ is defined as
$$\alpha\text{-}\stability_{f(D)}(t, k \mid Sb, \resChanges) = \frac{\vol(\resChanges \cap \outerZone(Sb))}{\vol(\resChanges)}$$
\end{definition}

\section{Computing Local Stability}
\label{sec:computing}

We propose \textbf{\algoName}, a sampling-based algorithm to estimate $\alpha$-local stability given a database $D$, a tuple $t\in D$, a ranking function $f$, a set of reasonable changes $\resChanges$, a value $k$, and parameters relating to the quality and confidence of the approximation which we discuss in the following. The algorithm relies on the fact that, despite the hardness of computing the stable zone boundary (\Cref{thm:hardness_bound_comp}), given a set of refinements $R$, we can efficiently bound $\alpha$, such that $R$ is an $\akStableZoneBound$ with high probability. 
We first show that with a sufficient number of samples, we can estimate the probability that a given set $R$ contains a $k$-unstable refinement. 

\begin{proposition}\label{prop:hoeffding-pac}
    Given a database $D$, a ranking function $f$, a value $k$, a tuple $t \in D$, a confidence level $\delta$, a confidence range $\eta$, and a set of refinements $R$, let 
    \begin{align*}
p &= \mathbb{P}_{\varepsilon \sim \mathcal{U}(\resChanges \cap \outerZone(R))}[ \varepsilon(t) \text{ is } k\text{-unstable}]\\
\widehat{p} &= \frac{1}{N}\sum_{\varepsilon \in \mathcal{S}} \mathbbm{1}[\varepsilon(t) \text{ is } k\text{-unstable}]
    \end{align*}
    where $\mathcal{S}$ is a set of $N = \frac{1}{2\eta^2} \ln \frac{1}{\delta}$ samples drawn i.i.d. uniformly from $\resChanges \cap \outerZone(R)$. Then we have that 
    $$\mathbb{P}[p - \widehat{p} \leq \eta] \geq 1 - \delta$$
\end{proposition}

\begin{sketch}
    This results as an immediate application of the Hoeffding inequality \cite{statisticallearningbook}.
\end{sketch}
The following is a corollary of \Cref{prop:hoeffding-pac}
\begin{corollary}\label{cor}
  $N = \frac{1}{2\eta^2} \ln \frac{1}{\delta}$ samples drawn i.i.d. uniformly from $\resChanges \cap \outerZone(R)$ are sufficient to confirm that $R$ is an $\akStableZoneBound$ for $\alpha = \widehat{p} + \eta$  with probability $1 - \delta$.    
\end{corollary}

\begin{figure}[t]
    \begin{subfigure}[t]{0.31\linewidth}
        \centering
        \includegraphics[width=1.66cm]{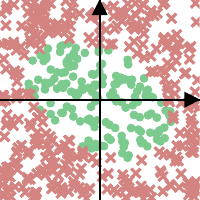}
        \caption{Construction}
        \label{fig:construction-sampling}
    \end{subfigure}
    \begin{subfigure}[t]{0.31\linewidth}
        \centering
        \includegraphics[width=1.66cm]{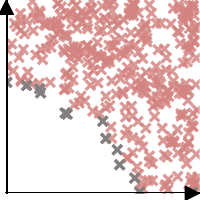}
        \caption{$Sb$}
        \label{fig:construction-sb}
    \end{subfigure}
    \begin{subfigure}[t]{0.31\linewidth}
        \centering
        \includegraphics[width=1.6cm]{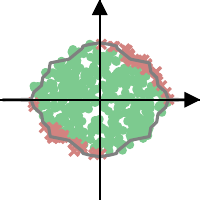}
        \caption{Verification}
        \label{fig:verification}
    \end{subfigure}
    \caption{Illustrated example of the steps of \algoName, based on the space of refinements shown in \Cref{fig:definitions}. (a) Green (circular) points represent sampled stable refinements, while red (cross) points represent unstable refinements. (b)  The projection of the $k$-unstable refinements in $C$ onto $\mathbb{R}^n_{\geq 0}$ ($U_C^+$). The gray (cross) points represent those along the estimated stable zone boundary. (c) Samples taken from the estimated stable zone  $\resChanges \cap \outerZone(Sb)$ for the verification step.}
    \label{fig:lstability-demo}
\end{figure}

$\algoName$ builds on this idea. It gets as input a database $D$, a tuple $t \in D$, a ranking function $f$, a value $k$, a set of reasonable changes $\resChanges$, a confidence level $\delta$, a sample budget $N$, a confidence range $\eta$, and consists of two phases. The first generates a set of refinements $Sb$, which is an $\akStableZoneBound$ with high probability based on \Cref{cor}. In the second phase, $\algoName$ outputs an estimation of the ratio between the volume of the estimated stable zone defined by $Sb$ and the volume of \resChanges. \Cref{fig:lstability-demo} provides an illustrative overview of \algoName.
 In the first phase, $\algoName$ computes a set of refinements $Sb$, which is an $\akStableZoneBound$ with high probability. This is done in two sample-based steps: set \emph{construction} and \emph{verification}. 
 In the construction step, illustrated in \Cref{fig:construction-sampling,fig:construction-sb}, the algorithm proceeds as follows:
\begin{enumerate}[leftmargin=*, nosep]
    \item Samples a set $C$ of $N$ refinements uniformly from \resChanges~and assigns to $U_C^+$ the projection of the $k$-unstable refinements in $C$ onto $\mathbb{R}^n_{\geq 0}$ (see \Cref{fig:construction-sampling}).
    \item Computes $Sb = Min_\prec(U_C^+)$  (see \Cref{fig:construction-sb}).
\end{enumerate}
In the verification step, illustrated in \Cref{fig:verification}, the algorithm performs the following:
\begin{enumerate}[leftmargin=*, nosep]
    \setcounter{enumi}{2}
    \item It samples a set $V$ of a sufficient number of refinements ($N = \frac{1}{2\eta^2} \ln \frac{1}{\delta}$ samples) uniformly from $\resChanges \cap \outerZone(Sb)$ 
    \item It computes the proportion of $k$-unstable refinements in $V$ to the overall number of samples $N$, thereby determining $\alpha$ such that $Sb$ is $\akStableZoneBound$ with probability $1-\delta$ based on \Cref{cor}
\end{enumerate}

Finally, in the second phase, given $Sb$ from the first phase, $\algoName$ applies standard Monte Carlo methods~\cite{montecarlo} to estimate $\alpha\text{-}\stability_{f(D)}(t,k\mid Sb,\resChanges)$. 
In cases where the stable zone is guaranteed to be convex (\eg when the ranking function is based on scores computed by a linear combination of attributes), alternative methods may be used as a drop-in replacement. In particular, \cite{fast-practical-cvx-polytope-volume-est,practical-volume-algo} describe efficient algorithms for estimating the volumes of convex polytopes with theoretical guarantees. While these methods offer stronger theoretical guarantees for both accuracy and efficiency when applicable, the Monte Carlo approach is generally efficient and sufficiently accurate for our purposes. Moreover, we note that the convexity assumption does not hold in many practical settings.

\section{Optimizations}
\label{sec:optimizations}

We next outline three optimizations to \algoName. The first reduces the set of reasonable changes \resChanges, allowing the algorithm to utilize the sample budget more effectively. 
The second optimization aims to reduce the cost of evaluating the ranking function $f$. This optimization is applicable when a modification of a tuple in the data does not affect the relative order of other tuples obtained by the ranking function, which is a common property of ranking functions. Finally, we suggest bounding $\alpha$ to reduce the running time. This is done by \emph{iteratively} running the construction and verification steps with a partial sample budget in each iteration, allowing early termination if a desired level of $\alpha$ has been reached.

\subsection{Reducing the Set of Reasonable Changes}
The first optimization aims at improving the utility of the sampling budget by narrowing down the set of reasonable changes such that it is guaranteed to still contain all $k$-stable refinements under the stable zone boundary. Using a smaller set of reasonable changes reduces the sample space for the construction step of \algoName, allowing for more informative sampling as the excluded refinements can not contribute to the computation of the stable zone boundary.

To reduce the set of reasonable changes, we use single-dimensional refinements, i.e., refinements that modify a single attribute. Considering such refinements allows us to eliminate refinements that can not be within the stable zone boundary. More formally, consider an individual attribute $a_i$, and a refinement $\varepsilon \in \mathbb{R}^n$ such that $|\varepsilon_{i}| > 0$, and $\varepsilon_{j} = 0$ for every $j \neq i$. Furthermore, assume that $\varepsilon(t)$ is $\unstable{k}{f}{D}$. As a consequence of the definition of $\kStableZoneBound$, we then have that any refinement $\varepsilon'$ such that $|\varepsilon_i| \leq |\varepsilon'_i|$ cannot be in $\outerZone(\kStableZoneBound)$.

\begin{example}
    Consider our running example and $\varepsilon^{max} = \langle 5, 3 \rangle$ defined in \Cref{ex:reasonable-changes}. \Cref{fig:reduced-reasonable-changes} depicts the stability of refinements for Dempster Uni. and $k = 0$. Consider the refinement $\varepsilon = \langle 0, 2 \rangle$ as shown in \Cref{fig:reduced-reasonable-changes}, $\varepsilon(t_2)$ is $\unstable{0}{f}{D}$. Thus, any refinement with more than $\pm 2$ systems publications cannot be in $0\text{-}\textsf{SB}_{f(D)}^t$. In fact, any refinement that lies above the gray cannot be part of the stable zone boundary.
\end{example}

\begin{figure}[t]
    \centering
    \includegraphics[width=0.58\linewidth]{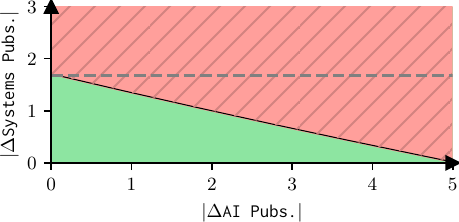}
    \caption{The stability of refinements for Dempster Uni. and $k = 0$. Refinements in the green (solid) area are $0$-stable and the red (hatched) area $0$-unstable. The reduced set of reasonable changes containing all refinements on the stable zone boundary consists of those under the gray dashed line.}
    \label{fig:reduced-reasonable-changes}
\end{figure}

To reduce \resChanges, a portion of the sampling budget is used to sample single-dimensional refinements $\varepsilon^{max}_i$ for every attribute $a_i$. Let $\varepsilon^*_i$ denote the minimal (absolute) value of $\varepsilon_i$ in the sample such that $\varepsilon(t)$ is $\unstable{k}{f}{D}$. We then define $\resChanges^*$ as the set of reasonable changes where $\varepsilon^{max}_i$ = $\varepsilon^*_i$. \algoName\ then uses $\resChanges^*$ instead of $\resChanges$, and given an estimation of $\alpha\text{-}\stability_{f(D)}(t, k \mid Sb, \resChanges^*)$, we scale it by a factor of $\frac{\vol(\resChanges^*)}{\vol(\resChanges)}$ to obtain $\alpha\text{-}\stability_{f(D)}(t, k \mid Sb, \resChanges)$.
When the ranking function is monotone \cite{monotonic}, reducing $\resChanges$ can be done efficiently using binary search, which we detail in \Cref{sec:apdx-binary-search}.

\subsection{Reduce Re-ranking Cost}
\label{sec:ranking-fewer-tuples}
To determine whether a refinement $\varepsilon$ is $k$-stable, \algoName\ needs to compute the value $\Delta_{f(D)}(t, \varepsilon(t))$, which in turns depends on the value of $f(D_{t \rightarrow \varepsilon(t)})[\varepsilon(t)]$. The function $f$ may be arbitrarily complex in its dependence on the size of the database $D$, e.g., in the case of learning to rank functions. Thus, evaluating $f(D_{t \rightarrow \varepsilon(t)})$ may be costly, in particular, when the sample budget $N$ is large. However, we note that determining whether $\Delta_{f(D)}(t, t') \leq k$ can be done without the cost of computing $f(D_{t \rightarrow \varepsilon(t)})$ for tuple-independent ranking functions which we next define.

\begin{definition}[Tuple-independent ranking function]
    \label{def:tuple-independent-ranking-function}
    A ranking function is \textbf{tuple-independent} if for any $t \in D$ and $t'$ obtained by refining $t$, and all $a, b \in D \setminus \{t, t'\}$ such that $f(D)[a] < f(D)[b]$, then $f(D_{t \rightarrow t'})[a] < f(D_{t \rightarrow t'})[b]$.
\end{definition}

Tuple independence is a common feature of ranking functions. For instance, the scoring function based on the geometric mean of attributes in \Cref{ex:running} is tuple-independent, since altering the attributes of one tuple only modifies the score of that tuple, leaving the relative ordering between the rest of the tuples unchanged.

Given a tuple $t$ and a refinement $\varepsilon$, let $t^\uparrow$ ($t^\downarrow)$ be the tuple $k + 1$ positions above (below) $t$ in $f(D)$ where $f$ is a tuple-independent ranking function. Since the relative ordering between all tuples except for the refined tuple $\varepsilon(t)$ remains the same, to determine whether $\Delta_{f(D)}(t, \varepsilon(t)) \leq k$, it is enough to check whether the relative order between $\varepsilon(t)$ and $t^\uparrow$ ($t^\downarrow$) is different than that of $t$ and $t^\uparrow$ ($t^\downarrow$). Therefore, we only need to evaluate $f$ over three tuples (a constant number), which can be significant when the database is large and the ranking function evaluation heavily depends on the data size (\eg in LtR models), as we show in our experimental evaluation (\Cref{sec:opt_exp}).

\subsection{Reduce Sampling for Bounded $\alpha$} 
\begin{figure}[t]
    \centering
    \includegraphics[width=\linewidth]{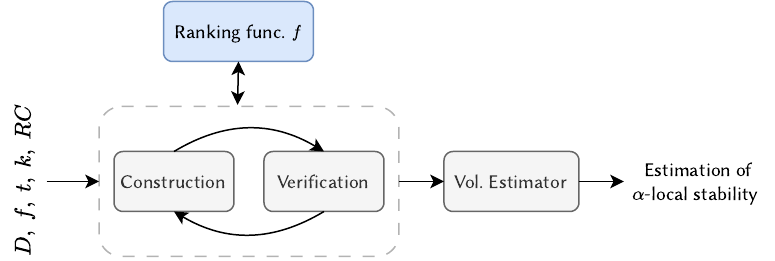}
    \caption{Overview of the components of \algoName}
    \label{fig:iterative}
\end{figure}

$\algoName$ returns the level of minimality $\alpha$ attained by the refinements boundary it computes. 
However, it is natural to opt for a desired $\alpha$, and it may be unnecessary for the construction step to sample as many times as it does in order to obtain a set of refinements $Sb$ such that $Sb$ is a $\akStableZoneBound$ with the desired $\alpha$. To this end, we extend $\algoName$ with the ability to apply iteratively the \emph{construction} and \emph{verification} steps as depicted in \Cref{fig:iterative}. In each iteration, the algorithm computes a set $Sb$ using a limited amount of samples, and then checks whether it is $\akStableZoneBound$ for $\alpha\leq$ the given bound. If so, the algorithm continues to the volume estimation of the $Sb$ and \resChanges. Otherwise, the next construction step merges the unstable refinements from the previous verification step into $Sb$ to improve its estimation using these counterexamples, and samples only from $\resChanges \cap \outerZone(Sb)$ which ensures the sampled refinements can only improve our estimation of $Sb$. The process runs for $\ell$ iterations, as determined by the user.

The choice of the number of samples taken per iteration by the construction step introduces a trade-off. The goal is to take just enough samples to get a good enough estimated stable zone boundary, namely, a set which is $\akStableZoneBound$ with the desired $\alpha$.  However, when more than a single iteration is required to obtain a set $Sb$ such that it is $\akStableZoneBound$ for the given $\alpha$, the verification step is evaluated multiple times. 
We propose apportioning the budget equally to each iteration, taking into account the samples taken by the verification step in each iteration. Specifically, if $N$ samples were allocated to the construction step and $V$ samples to the verification step in the vanilla version of \algoName, we allocate $\frac{N+V}{\ell} - V$ samples to the construction step in each iteration. Under this allocation, we take as many samples as was originally budgeted in the worst case.

\section{Detecting Dense Regions}
\label{sec:detecting-dense-regions}

\begin{figure}
    \begin{subfigure}[c]{0.48\linewidth}
        \centering
        {\small
        \begin{tabular}{cc:c}
        \hline
            \textbf{X} & \textbf{Y} & \textbf{Score $\downarrow$} \\ \hline
            \rowcolor{ClusterBlue!50} 9.71 & 0.86 & 10.57 \\
            \rowcolor{ClusterGreen!50} 1.38 & 9.15 & 10.53 \\
            \rowcolor{ClusterBlue!50} 9.30 & 1.20 & 10.51 \\
            \rowcolor{gray!20} ... & ... & ... \\
            \rowcolor{ClusterBlue!50} 8.31 & 0.52 & 8.83 \\
            \hline
        \end{tabular}
        }
        \caption{Scores of dataset}
        \label{tab:attribute_cluster_scores}
    \end{subfigure}%
    \hfill%
    \begin{subfigure}[c]{0.48\linewidth}
        \centering
        \includegraphics[clip,trim={0.5cm 0.5cm 0.5cm 0.5cm},width=0.6\linewidth]{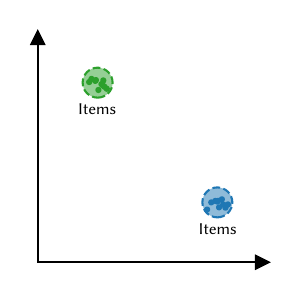}
        \caption{Dataset visualized}
        \label{fig:attribute_cluster_graph}
    \end{subfigure}%
    \caption{Example showing two distant clusters in the space of attributes can map to similar scores, suggesting they are all part of a single dense region in the ranking}
    \label{fig:attribute_clustering}
\end{figure}

The local stability of a tuple $t$ is defined with respect to a given $k$, defining the range of positions with respect to which local stability is evaluated. This parameterization allows our definition of local stability to be useful in the presence of \emph{dense regions} in the ranking. Thus far, we have focused on computing the local stability of a given tuple in a ranking. In our running example, the local stability measure evaluates how well a particular university aligns with its assigned position in the ranking. An interesting alternative perspective is to explore, for a given university, the potential bounds on its possible position within the ranking, based on the local stability measure. Namely, identifying the dense region around a tuple in the ranking, if it is in one.

Intuitively, in order to identify dense regions, we need to cluster items in the ranking with similar utilities. However, this may be impossible when the ranking is not score-based (\eg some list-wise LtR models \cite{STARank}), so there is no meaningful measure of utility for a given item. Another na\"{i}ve solution involves clustering based on the tuples' attributes. However, this approach assumes that all tuples ranked similarly must be similar in their attributes, which may not always be the case, as we next illustrate.
\begin{example}
\Cref{tab:attribute_cluster_scores} shows items consisting of two attributes $x$ and $y$ ranked in descending order of $x+y$, each represented as a point in \Cref{fig:attribute_cluster_graph}. As demonstrated on \Cref{fig:attribute_cluster_graph}, two clusters may be observed based on their attribute values depicted as blue and green dashed circles in the figure. However, their scores, as shown in \Cref{tab:attribute_cluster_scores}, are all similar, indicative of all items belonging to a single dense region. 
\end{example}

\begin{algorithm}[t]
\small
\SetAlgoLined
\DontPrintSemicolon
\KwIn{Database $D$, a tuple $t \in D$, a ranking function $f$, a set of reasonable changes $\resChanges$ and a sample budget $N$}
\KwOut{A value $k$ (the size of the dense region of $t$)}
$S \gets \emptyset$\;\label{line:firstPartStart}
\For{$i \gets 1$ \KwTo $N$}{
    $\varepsilon \gets \mathsf{Sample}(\mathcal{U}(\resChanges))$\;
    $t' \gets \varepsilon(t)$\;
    $S \gets S \cup \{(\varepsilon, \Delta_{f(D)}(t, t')) \}$\;
}\label{line:sampleEnd}

$k^* \gets \max(\{p \mid (\varepsilon, p) \in S \})$ \Comment*{Largest change of positions induced by samples from $\resChanges$}\label{line:max-k}
$\widehat{\stability}(t, k^*) \gets 1$\;  \label{line:max-k_stability}
\For{$k \gets 0$ \KwTo $k^*-1$}{\label{line:startEstLoop}
$S \gets S \setminus \{(\varepsilon, p) \in S \mid {p \leq k \wedge (\not\exists (\varepsilon', p') \in S,~~(p' > k) \wedge (\varepsilon' \preceq \varepsilon))} \}$\;\label{line:removeFromS}
$\widehat{\stability}(t, k) \gets 1 - \frac{|S|}{N}$\; \label{line:estimateStability}
}\label{line:firstPartEnd}
$d_0 \gets \widehat{\stability}(t, 0)$\;\label{line:secondPartStart}
\For{$k \gets 1$ \KwTo $k^*$}{
$d_k \gets \widehat{\stability}(t, k) - \widehat{\stability}(t, k - 1)$\label{line:stbDiff}
}
$C_s, C_l \gets \textsf{Cluster}(d_0, \dots, d_k)$ \Comment*{$\textsf{Cluster}$ returns a partitioning of $d_0, \dots, d_k$ into two groups: small and large}\label{line:cluster}
\Return $\min_{k \in [k^*], d_k \in C_l} k$\label{line:secondPartEnd}
\caption{\denseRegionAlgo}
\label{alg:detect}
\end{algorithm}

To this end, we propose \textbf{\denseRegionAlgo}, a heuristic for computing a value $k$ for a given tuple in the database, following the intuitive meaning of dense regions viewed through the lens of our definition of local stability.
The high-level idea of \denseRegionAlgo is to compare the differences in the stability values of a given tuple for different $k$ values and identify the last $k$ value before a large difference occurs. The pseudocode of \denseRegionAlgo is shown in Algorithm~\ref{alg:detect}. The algorithm first estimates the local stability for different values of $k$ (lines~\ref{line:firstPartStart}-\ref{line:firstPartEnd}). Then (lines~\ref{line:secondPartStart}-\ref{line:secondPartEnd}), the algorithm uses clustering to partition the differences in the computed stability values into small and large changes, and returns as suggested $k$ the position where the first large difference was observed. We next provide details on each phase.

\paragraph{\textbf{Estimating local stability for multiple values of $k$}}
To determine the first value of $k$ for which we observe a large difference in estimated local stability, we first need to estimate the local stability for different values of $k$. The challenge here is to do so efficiently. Evaluating \algoName for all possible values of $k$ is unnecessarily expensive when we only need rough estimates of the local stability. 

In order to do so, \denseRegionAlgo samples refinements according to a given sample budget $N$ and stores each sampled refinement $\varepsilon$ together with the change in position $\Delta_{f(D)}(t, \varepsilon(t))$ in a set of samples $S$ (lines~\ref{line:firstPartStart}-\ref{line:sampleEnd}). After sampling, the algorithm determines the maximum change in positions induced by any of the samples and denoted by $k^*$ (line~\ref{line:max-k}). These samples are used to estimate the local stability for every $k$ between $0$ and $k^*$, denoted by $\widehat{\stability}(t, k)$. $k^*$ is assumed to be the maximum change in positions induced by \emph{any} refinement contained in the set of reasonable changes $\resChanges$ (which is true for a sufficiently large number of samples), thus the estimated stability for $k^*$ is $1$ (line~\ref{line:max-k_stability}). 
For every value of $k$ from $0$ to $k^*-1$ the stability is estimated using the set of samples $S$. Intuitively, the stability value for a given $k$ can be estimated as the number of samples $(\varepsilon, p)$ in $S$ such that $\varepsilon$ is $k$-stable (namely, $p\leq k$), and there is no $(\varepsilon', p')\in S$ such that $\varepsilon'\preceq\varepsilon$ and $\varepsilon'$ is $k$-unstable (namely, $p' > k$), divided by the total number of samples $N$. 
Note that the estimated $(k+1)$-stable zone boundary is above the estimated $k$-stable zone boundary. 
Thus, the estimation can be done iteratively (lines~\ref{line:startEstLoop}-\ref{line:firstPartEnd}) as follows: for each $k$, we remove from $S$ the $k$-stable refinements not containing any sampled $k$-unstable refinement (line~\ref{line:removeFromS}) and estimate the local stability for $k$ as $1-\frac{|S|}{N}$ (line~\ref{line:estimateStability}).

\paragraph{\textbf{Clustering-based suggestion of dense region}}
Given an estimation of the local stability for the observed values of $k$, according to our notion of a dense region, the extent of the dense region in which the tuple lies is the first value of $k$ for which the succeeding value of $k$ has a significantly larger value of local stability. Therefore, \denseRegionAlgo computes the difference in stability (lines~\ref{line:secondPartStart}~-~\ref{line:stbDiff}) for each value of $k$, $d_k$, as the local stability for $k$ minus the local stability for $k-1$ (and fixing $d_0 = \widehat{\stability}(t, 0)$). It then employs a Fisher-Jenks natural breaks~\cite{jenks} clustering to partition the differences into two groups, $C_s$ and $C_l$ of small and large differences, respectively (lines~\ref{line:cluster}). Finally, \denseRegionAlgo returns $k$, where $k$ is the smallest value such that $d_k$ is in the set of large differences $C_l$ (line~\ref{line:secondPartEnd}).

\section{Experiments}
\label{sec:experiments}

In this section, we present an experimental evaluation. We start with two case studies whose goal is to demonstrate the usefulness of our local stability definition. We then study the effect of the sample budget $N$, the confidence level $\delta$, and the confidence range $\eta$ on the performance of \algoName. We examine the performance of \denseRegionAlgo using real and synthetic data, and show the scalability of the algorithm and the usefulness of our proposed optimizations. We conclude with a comparison to the definition of stability that was presented in~\cite{AJMS18}. 

\subsection{Experimental Setup}
\label{sec:exSetup}

\paragraph{\textbf{Implementation details}} 
Our experiments were evaluated on Google Cloud Compute Engine \texttt{c4d-highcpu-4} servers, with 4 vCPUs (AMD EPYC 9005 series) and 7 GB of memory.
We implement the algorithms in Python 3.13, using Polars\footnote{\url{https://pola.rs}} 1.29 to load and manipulate data. We repeat each experiment $10$ times, and show the averages of the quantities of interest (\eg runtime, estimated ($\alpha$-)local stability) in the figures.

We implement the uniform sampling over the refinements under the estimated boundary $Sb$, used in the verification and optimized construction step of \algoName, by rejection sampling. 
More concretely, we sample uniformly from $\resChanges$, and reject refinements containing any refinement of $Sb$. 
The efficiency of rejection sampling can be poor when the volume of the area we want to sample from is relatively small compared to $\resChanges$. 
Therefore, we use a threshold over the estimated volume $\tau_v$, such that when the estimated ratio between the volume of the area and its containing hyperrectangle is below $\tau_v$ at the end of the construction step, we skip the verification step and stop early. We use $\tau_v = 0.05$, as improving the estimation below this point yields diminishing returns with respect to the verification cost. We evaluated 10,560 experiments in total, where the verification step was skipped due to $\tau_v$ in 7.6\% of them.

We evaluate \algoName with all optimizations described in \Cref{sec:optimizations} enabled and compare its performance to the basic version of \algoName without optimizations as presented in \Cref{sec:computing}.

\paragraph{\textbf{Parameters setting}}
For the construction step, we set the total number of samples to be at most 750,455, where in each iteration 20,000 samples are taken when the maximal number of iterations $\ell$ is $20$.
We chose this number empirically from the sample budget parameter study in \Cref{sec:performance-vs-parameters}. For the verification step, unless otherwise specified, we set the bound on $\alpha$ to be $0.05$ and the confidence bound $\delta = 0.05$. That is, we are satisfied if we are 95\% sure that 95\% of refinements contained in the estimated boundary are $k$-stable, for the given value of $k$. In order to have a decent chance of attaining $\alpha \leq 0.05$, we set the confidence range $\eta = 0.01$.

\paragraph{\textbf{Datasets \& ranking functions}}
We evaluated our algorithms using both real and synthetic datasets, utilizing different types of ranking functions as follows.

\textbf{NBA (2023-2024):} We consider a ranking of the top-100 players according to a popular sports website \cite{nba-ranking}. Since the ranking methodology is not provided, we learn a ranking function from the reference ranking using player statistics from the 2023-2024 season (collected from \cite{nba-stats}) as the input features. We use LightGBM \cite{lightgbm} to learn a regressor based on the same five statistics used in \cite{whynotyet}: points (\texttt{PTS}), total rebounds (\texttt{TRB}), assists (\texttt{AST}), steals (\texttt{STL}), and blocks (\texttt{BLK}). For these five statistics, we consider the season totals since we are interested in ranking overall season performance. By default, we let the set of reasonable changes ($\resChanges$) be approximately $5\%$ of the maximum difference in each attribute. 

\textbf{CSRankings:} We use a dataset and ranking function from \cite{sharp} (based on \cite{CSRankings}), comprising 188 universities and the counts of their faculty and publications for four different fields of computer science: artificial intelligence (\texttt{AI}), systems (\texttt{Sys}), theory (\texttt{Thry}), and interdisciplinary research (\texttt{Intdsc}). 
 The count of publications is adjusted to account for the number of authors in each publication by dividing the credit equally, \eg a paper with $2$ authors from different institutions counts as $0.5$ publications for each of the involved institutions (\ie fractional changes in the publication count for a given institution are valid and meaningful). The universities are ranked in descending order of their scores, which are computed as a variation of 
 geometric mean over the adjusted publication counts in each area--specifically,
$$\varphi(t) = \sqrt[27]{(t.\texttt{AI} + 1)^{5} \cdot (t.\texttt{Sys} + 1)^{12} \cdot (t.\texttt{Thry} + 1)^{3} \cdot (t.\texttt{Intdsc} + 1)^{7}}$$
We let $\resChanges$ be approximately $5\%$ of the maximum difference in each attribute--about $\pm 4$ adjusted publication count in \texttt{AI} and $\pm 1$ adjusted publication count in the other fields. 

\textbf{Synthetic:} We generate synthetic datasets with dense regions specifically in mind. In order to generate a synthetic dataset of $N$ tuples and $d$ attributes, we iteratively generate dense regions with a random number of tuples until $N$ tuples are in the dataset. Each dense region is assigned a score $s$, which is separated from the preceding dense region by a constant margin $c$. We then draw tuples from a multivariate Gaussian distribution whose mean is the one-vector scaled by $\frac{s}{d}$, and the covariance matrix is a scaled identity matrix. We rank the resulting dataset by the ranking function, sorting the tuples in descending order of the sum of their attributes. We let $\resChanges$ be the vector whose components are proportional to $\frac{c}{d}$, ensuring that local stability values are comparable across varying attribute counts for a fixed margin $c$.


\subsection{Case Studies}
\label{sec:case-studies}
\begin{figure}[t!]
    \begin{subfigure}[t]{\linewidth}
        \centering
        \includegraphics[clip,trim={0.25cm 0.25cm 0.25cm 0.25cm},width=\linewidth]{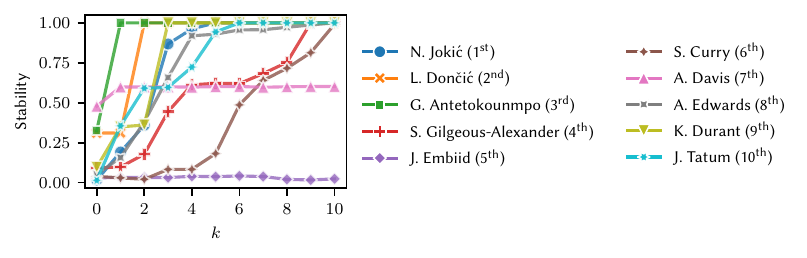}
        \caption{Effect of $k$ on stability values}
        \label{fig:nba-stability-vs-k}
    \end{subfigure}
    \begin{subfigure}[t]{\linewidth}
        \centering
        \includegraphics[clip,trim={0.25cm 0.25cm 0.25cm 0.25cm},width=\linewidth]{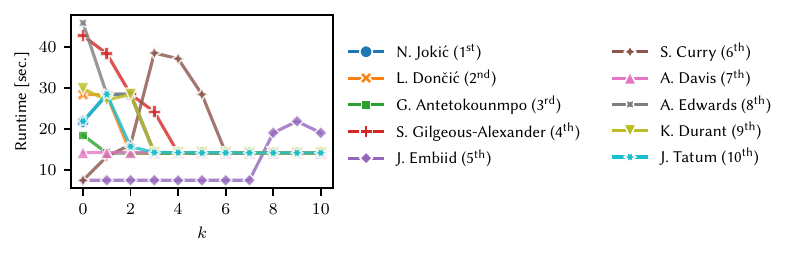}
        \caption{Effect of $k$ on runtime}
        \label{fig:nba-runtime-vs-k}
    \end{subfigure}
    \caption{Case study results for top-10 NBA players}
    \label{fig:nba-varying-k}
\end{figure}

We start with two case studies on the NBA and CSRankings datasets.
\subsubsection{Case Study: NBA Player Rankings}\label{sec:nba-case-study} Rankings are frequently featured in sports, and can drive decision-making in a number of contexts, \eg in fantasy leagues, drafts, and estimating the strength of a given team. In this case study, we investigate the local stability of NBA players in the 2023-2024 season as ranked by a learned ranking function, with the goal of ranking players by their performance that season. Since the top of the ranking receives significantly more attention \cite{power-of-rankings,ecommerce-ranking}, we focus on the local stability of the top-10 ranked players.
The local stability of the top-10 players for different values of $k$ is depicted in~\Cref{fig:nba-stability-vs-k}.

Assume we wish to choose the season's most valuable player (MVP). According to the ranking, a natural choice would be the top-ranked player of this season, Nikola Joki\'{c}. However,~\Cref{fig:nba-stability-vs-k} shows that his ranking in 1\textsuperscript{st} is very unstable with local stability of $0.02$ for $k=0$. Indeed, with minor modifications to his total statistics that season, \eg $+10.04$ in \texttt{PTS}, $-3.04$ in \texttt{TRB}, $-0.56$ in \texttt{AST}, $-1.01$ in \texttt{STL}, and $-1.11$ in \texttt{BLK}, he would be ranked 2\textsuperscript{nd}. Notice that these changes are particularly small in comparison to his overall season totals: in this season, Joki\'{c} had 2,085 \texttt{PTS}, 976 \texttt{TRB}, 708 \texttt{AST}, 108 \texttt{STL}, and 68 \texttt{BLK}. Furthermore, his ranking among the top-3 is also unstable, with local stability of $0.36$ for $k=2$. This suggests that naming Joki\'{c} as this season's MVP may not be well-founded under this ranking function. Based on this ranking function, we may instead conclude that choosing Luka Don\v{c}i\'{c} is more justifiable given that he is certainly in the top-4 (more likely than Joki\'{c}). Alternatively, we may refine the ranking function (\eg by considering only a subset of the attributes) to obtain a ranking with higher local stability values that can be used to justify the chosen MVP. 

We observed that most of the top-10 players are not stable within their own position: the highest estimated stability when $k=0$ is $0.46$, in the case of Anthony Davis.
While most players are unstable for $k=0$, we observed that the majority of players are stable within $\pm 3$ ranks of their initial ranking. Therefore, we may infer this particular ranking function is reasonably locally stable for the players in the top-10, and confirm that it is overall \emph{well-founded}: namely, highly-ranked players remain highly-ranked even under small modifications in their statistics for this season. 
A notable exception is Joel Embiid, whose stability is very low for all evaluated values of $k$. This suggests that the learned ranking function has overfit to Embiid's statistics, as even very small modifications to his statistics cause him to drop out of the top-10 players. This interpretation is supported by the fact that Embiid played far fewer games than other players in the 2023-2024 season (only 39, compared to more than 70 for most other players), due to injuries. Therefore, his total statistics for the season are much lower than they would have been otherwise, making them more similar to lower-ranked players. 

\Cref{fig:nba-runtime-vs-k} shows the runtime of (the optimized) \algoName for different values of $k$ for each of the top-$10$ players.
Broadly speaking, the runtime is decreasing for increasing values of $k$. Intuitively, this is explained by the fact that there are necessarily more $k$-unstable refinements for lower values of $k$, often resulting in more iterations taking place in \algoName to produce a good enough estimate. 
Notable exceptions are Stephen Curry and Joel Embiid, in which the runtime starts low (because they are below the estimation volume threshold $\tau_v$), rises when just above $\tau_v$, and then decreases again once their local stability increases. When compared to the basic version of \algoName, as presented in \Cref{sec:computing}, our optimized algorithm is $51.6\times$ faster in the best case, and $25.4\times$ faster on average. We show the full results in \Cref{fig:varying_k_noopts} in \Cref{sec:apdx-addtl-expts}.


\subsubsection{Case Study: CSRankings}
\label{sec:csrankings-case-study}
\begin{figure}[t!]
    \begin{subfigure}[t]{\linewidth}
        \centering
        \includegraphics[clip,trim={0.25cm 0.25cm 0.25cm 0.25cm},width=0.9\linewidth]{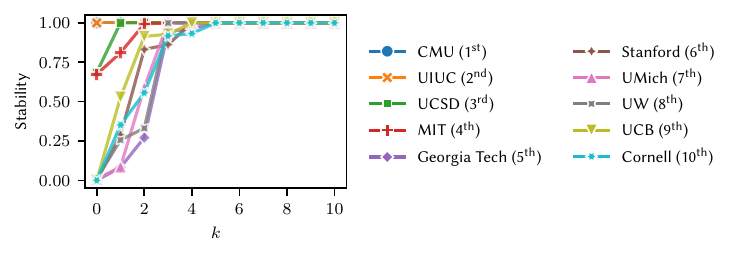}
        \caption{Effect of $k$ on stability values}
        \label{fig:csrankings-stability-vs-k}
    \end{subfigure}
    \begin{subfigure}[t]{\linewidth}
        \centering
        \includegraphics[clip,trim={0.25cm 0.25cm 0.25cm 0.25cm},width=0.9\linewidth]{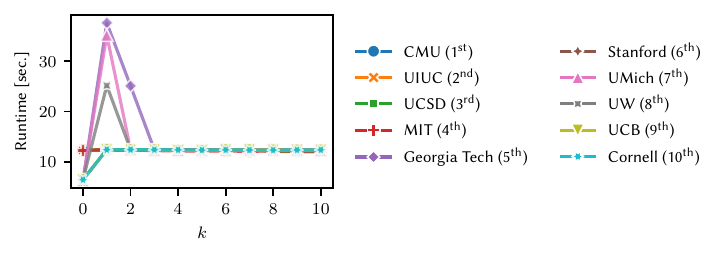}
        \caption{Effect of $k$ on runtime}
        \label{fig:csrankings-runtime-vs-k}
    \end{subfigure}
    \caption{Case study results for top-10 universities}
    \label{fig:csrankings-varying-k}
\end{figure}

Rankings also feature prominently in higher education, and have been the subject of significant study in order to develop methodologies that accurately reflect the perceptions of the quality of the institutions when compared with each other \cite{nrc-ranking,goto-rankings-considered-helpful}. We consider evaluating the local stability of institutions in a more specialized ranking, CSRankings \cite{CSRankings}, which ranks institutions based on publication metrics in computer science conferences. Given the focus placed on high positions in the ranking, we again consider the local stability of the top-10 universities.

\Cref{fig:csrankings-stability-vs-k} shows that the ranking is substantially locally stable. In fact, the top 2 universities, Carnegie Mellon University (CMU) and the University of Illinois at Urbana-Champaign (UIUC), are estimated to be completely locally stable: no change made to them from the set of reasonable changes causes them to rank outside of 1\textsuperscript{st} or 2\textsuperscript{nd}, respectively. For the remaining universities, we see that for $k=3$, all of the universities have an estimated local stability of more than $\frac{1}{2}$.
 Furthermore, from $k \geq 5$, all of the universities are estimated to be completely locally stable. This lends credence to the claim of these universities fielding the top-ranked computer science departments, since small changes in the ranking do not result in drastic changes in their position. 

\Cref{fig:csrankings-runtime-vs-k} shows the runtime of (the optimized) \algoName depending on the value of $k$.
 As expected, the runtime in general decreases with $k$. 
 As in the NBA case study, there are universities with estimated stabilities just above the threshold $\tau_v$, leading to higher runtimes for those values of $k$. We observed a speedup of up to $35.2\times$ in the best case, and of $19.1\times$ on average compared to the performance of the basic version of \algoName. We show the full results in \Cref{fig:varying_k_noopts} in \Cref{sec:apdx-addtl-expts}.

\subsection{Effect of Parameters}\label{sec:ex-paremeters}

The next set of experiments study the effect of the sampling budget $N$ and verification parameters $\delta$ (confidence level) and the $\eta$ (confidence range) on the running time of \algoName and the computed stability values. The observed effect of the parameters on the stability values were negligible, and thus are shown only in \Cref{sec:apdx-addtl-expts}.
For each dataset, we set $k$ such that a wide variety of values of local stability are observed across different tuples. Specifically, for NBA we chose $k=5$, and $k=1$ for CSRankings.

\label{sec:performance-vs-parameters}

\paragraph{\textbf{Sample budget}}
\begin{figure}[t]
    \begin{subfigure}[t]{\linewidth}
        \centering
        \includegraphics[clip,trim={0.25cm 0.25cm 0.25cm 0.25cm},width=\linewidth]{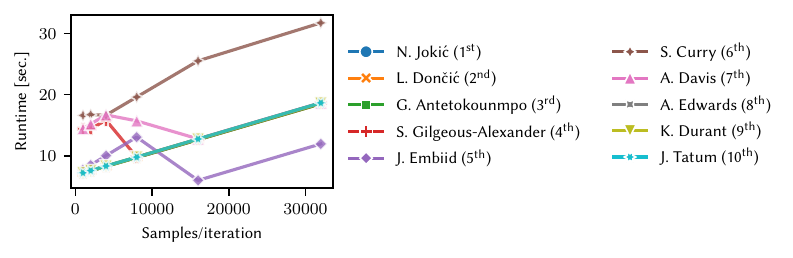}
        \caption{NBA}
        \label{fig:nba-runtime-vs-sample-budget}
    \end{subfigure}
    \begin{subfigure}[t]{\linewidth}
        \includegraphics[clip,trim={0.25cm 0.25cm 0.25cm 0.25cm},width=0.9\linewidth]{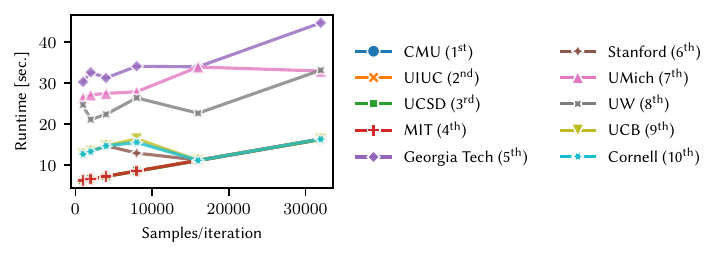}
        \caption{CSRankings}
        \label{fig:csrankings-runtime-vs-sample-budget}
    \end{subfigure}
    \caption{Effect of the number of samples per round taken by the construction step on runtime}
    \label{fig:runtime-vs-sample-budget}
\end{figure}

\Cref{fig:runtime-vs-sample-budget} depicts the running time for increasing sampling budget $N$ between 1,000 and 32,000. We observed a tradeoff in running time between small and large per-iteration sample budgets.
When the number of samples in a single iteration is insufficient to achieve the desired $\alpha$, multiple iterations are required, incurring the cost of multiple executions of the verification step. This phenomenon is more pronounced for tuples with low values of local stability (\eg Cornell in \Cref{fig:csrankings-runtime-vs-sample-budget}). 

\paragraph{\textbf{Verification parameters ($\eta, \delta$)}}

\begin{figure}[t]
    \begin{subfigure}[t]{\linewidth}
        \centering
        \includegraphics[clip,trim={0.25cm 0.25cm 0.25cm 0.25cm},width=\linewidth]{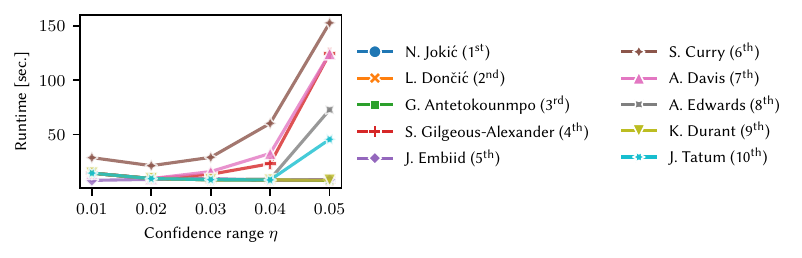}
        \caption{NBA}
        \label{fig:nba-runtime-vs-eta}
    \end{subfigure}
    \begin{subfigure}[t]{\linewidth}
        \includegraphics[clip,trim={0.25cm 0.25cm 0.25cm 0.25cm},width=0.9\linewidth]{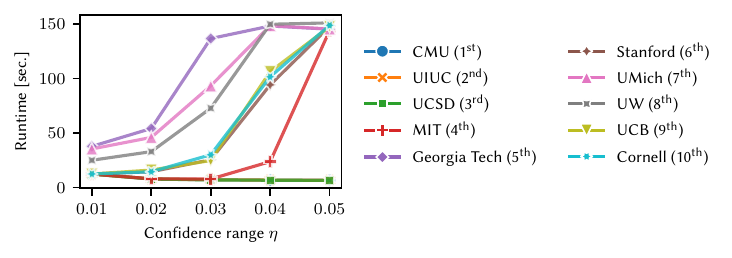}
        \caption{CSRankings}
        \label{fig:csrankings-runtime-vs-eta}
    \end{subfigure}
    \caption{Effect of the confidence range $\eta$ on runtime}
    \label{fig:runtime-vs-eta}
\end{figure}

\begin{figure}[t]
    \begin{subfigure}[t]{\linewidth}
        \centering
        \includegraphics[clip,trim={0.25cm 0.25cm 0.25cm 0.25cm},width=\linewidth]{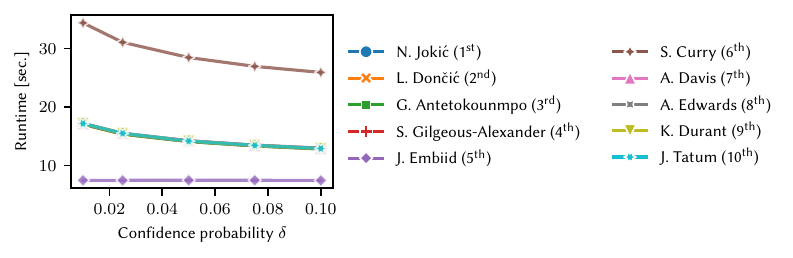}
        \caption{NBA}
        \label{fig:nba-runtime-vs-delta}
    \end{subfigure}
    \begin{subfigure}[t]{\linewidth}
        \includegraphics[clip,trim={0.25cm 0.25cm 0.25cm 0.25cm},width=0.9\linewidth]{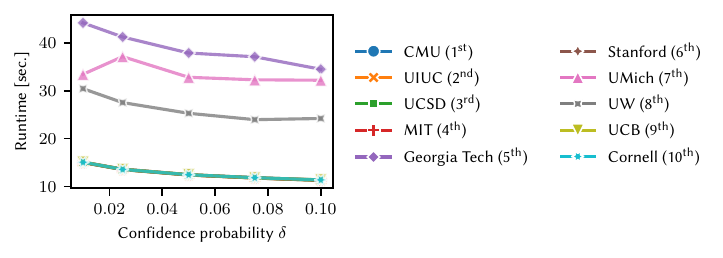}
        \caption{CSRankings}
        \label{fig:csrankings-runtime-vs-delta}
    \end{subfigure}
    \caption{Effect of the confidence probability $\delta$ on runtime}
    \label{fig:runtime-vs-delta}
\end{figure}

To assess the effect of $\eta$, we vary its value from $0.01$ to $0.05$. The results are shown in \Cref{fig:runtime-vs-eta}. We observed an increase in running times. This can be explained as follows. Recall that $\alpha = \widehat{p} + \eta$ where $\widehat{p}$ is the estimated probability that the computed $Sb$ contains an unstable refinement as defined in \Cref{prop:hoeffding-pac}. Thus, when $\alpha$ is fixed, as $\eta$ increases, the acceptable upper bound on $\widehat{p}$ decreases, meaning the portion of sampled unstable refinements under $Sb$ should decrease. Achieving a lower $\widehat{p}$ may require an increase in the number of samples, and as a result growing number of iterations, which incurs higher computation time.
Finally, \Cref{fig:runtime-vs-delta} shows the running time for values of $\delta$ from $0.01$ to $0.1$. As expected, the runtime decreases with increasing $\delta$ values, \ie for low confidence level, the running time is lower. 


\subsection{Dense Region Detection}
\label{sec:expts-dense-regions}
We next evaluate \denseRegionAlgo\ in terms of accuracy and runtime. To this end, we use the CSRankings dataset and the Synthetic dataset to assess the quality of the output given the ground truth over dense regions in the data. We do not consider the NBA dataset in this experiment, as the learned ranking function treats items with similar statistics in a substantially disparate manner (\eg the case of Joel Embiid in the NBA dataset), leading to the absence of clear dense regions in the rankings.

\Cref{tab:csrankings-top10} presents the top-10 universities in CSRankings along with their computed score. We observe a clear separation between different groups in the data based on their scores, which is highlighted in the table. Specifically, there is a notable gap in the scores for tuples ranked 1\textsuperscript{st}-4\textsuperscript{th}, whereas tuples ranked  5\textsuperscript{th}-8\textsuperscript{th} and tuples ranked  9\textsuperscript{th}-10\textsuperscript{th} have a more comparable score to each other, forming dense regions in the data. \emph{E.g.,} with only one less adjusted publication count in systems, Georgia Tech would be ranked 8\textsuperscript{th}.

The evaluation of \denseRegionAlgo\ was able to identify the dense regions. In particular, the $k$ value computed for each one of the first four tuples was $0$, indicating their stability with respect to their surrounding tuples. For Stanford, the $k$ value
computed by \denseRegionAlgo\ was $1$ while one might
 reasonably expect $k=2$ so as to include the University of Washington (UW). Indeed, the local stability value for $k=0$ is $7 \cdot 10^{-6}$, for $k=1$ is $0.29$ and for $k=2$ is $0.83$, \ie an increase of $0.29$ from $k = 0$ to $k=1$, while the increase from $k=1$ to $k = 2$ is $0.54$, which is significantly larger. However, the former is classified as being large by the clustering algorithm due to the smaller differences between other values of $k$ (which are all less than $0.15$). The observed $k$ values for the rest of the tuples matched the dense regions as they are depicted in \Cref{tab:csrankings-top10}. These results highlight the usability of \denseRegionAlgo\ to detect dense regions in the rankings as a reasonable approach to solving the problem. When accounting for dense regions, the local stability of the ranking becomes evident. Specifically, when choosing $k$ to span the dense area of the universities, we find that all the universities (except Cornell) have a local stability of at least $\frac{1}{2}$. In particular, Cornell has lower stability within its dense region than UC Berkeley (UCB) since UCB is buoyed by their relative strength in systems publications while Cornell is not, therefore making it feasible for Cornell to rank 11\textsuperscript{th}. 
 
\begin{table}[t]
    \caption{Top-10 computer science departments according to CSRanking \cite{sharp}, with dense regions highlighted}
    \scalebox{0.83}{
        \small
        \begin{tabular}{ccrrrr:r}
        \hline
        & \textbf{University} & \textbf{AI} & \textbf{Sys} & \textbf{Thry} & \textbf{Intdsc} & \textbf{Score $\downarrow$} \\
        \hline
        1\textsuperscript{st} & CMU & 71.4 & 11.9 & 21.1 & 13.8 & 19.53 \\
        2\textsuperscript{nd} & UIUC & 46.1 & 12.6 & 16.0 & 7.2 & 15.39 \\
        3\textsuperscript{rd} & UCSD & 31.6 & 9.0 & 10.1 & 10.3 & 13.00 \\
        4\textsuperscript{th} & MIT & 28.1 & 8.6 & 16.2 & 7.9 & 12.33 \\
        \rowcolor{CornflowerBlue!20}5\textsuperscript{th} & Georgia Tech & 28.5 & 7.8 & 6.9 & 10.2 & 11.58 \\
        \rowcolor{CornflowerBlue!20}6\textsuperscript{th} & Stanford & 36.7 & 5.4 & 13.3 & 11.5 & 11.56 \\ 
        \rowcolor{CornflowerBlue!20}7\textsuperscript{th} & UMich & 30.4 & 9.0 & 9.3 & 5.9 & 11.26 \\
        \rowcolor{CornflowerBlue!20}8\textsuperscript{th} & UW & 28.0 & 6.2 & 12.2 & 10.0 & 11.13 \\
        \rowcolor{Peach!20}9\textsuperscript{th} & UCB & 23.2 & 7.4 & 15.9 & 6.4 & 10.69 \\
        \rowcolor{Peach!20}10\textsuperscript{th} & Cornell & 42.0 & 5.7 & 12.8 & 6.8 & 10.66 \\
        \bottomrule
        \end{tabular}
    }
    \label{tab:csrankings-top10}
\end{table}

To further evaluate the performance of \denseRegionAlgo, we utilized the Synthetic dataset with $d=2$ and $100$ tuples, which was generated with dense regions and thus provides us with ground truth for $k$ values that cover the dense region to which the tuple belongs.
We executed \denseRegionAlgo\ on all $100$ tuples of the Synthetic dataset and compared the output $k$ value for each tuple with the ground truth value of $k$. 
For all $100$ tuples, \denseRegionAlgo\ was correctly able to recommend the value of $k$ that fits the dense region according to the data generation process. 

\paragraph{\textbf{Efficiency of \denseRegionAlgo}}
To show the efficiency of \denseRegionAlgo, we compared the runtime of \denseRegionAlgo\ to that of a variant where the stability values are computed using \algoName.  
The evaluation was done on all $100$ tuples of the Synthetic dataset. The total runtime for \denseRegionAlgo was 7 minutes, compared to 141 minutes when using \algoName for the local stability values computation. We observed a $20.3\times$ improvement overall, where \denseRegionAlgo was up to $32\times$ faster with an average of $20.4\times$ improvement.


\subsection{Scalability \& Optimizations}\label{sec:opt_exp}
We next examine the scalability of the optimized version of \algoName (LSt) compared to the basic version presented in \Cref{sec:computing} (Basic), and quantify the effect of our proposed optimizations.

\paragraph{\textbf{Scalability}}
We evaluate the scalability of our methods with growing data sizes and numbers of attributes using synthetic datasets. \Cref{fig:runtime-vs-datasize} shows the runtime of both the optimized and basic versions of \algoName as a function of the data size. 
As expected, LSt outperforms Basic. We observed a linear increase in the Basic's running time, with increasing data size, while the impact on LSt's running time was negligible. This confirms that the dependence on data size can be eliminated for tuple-independent ranking functions by the optimization to reduce the re-ranking cost.

The runtime as a function of the number of attributes is depicted in \Cref{fig:runtime-vs-num-attrs}. We note that the running time can be affected by the stability of the tuples, \ie tuples with stability below the threshold $\tau_v$ typically incur shorter runtimes due to early termination. Due to the randomness involved in the data generation process and tuple selection, it is likely that each evaluated dataset has a mixture of high and low stabilities.
Consequently, we report the maximum observed running time per tuple. For Basic, the runtime increases linearly with the number of attributes as each additional attribute incurs more overhead in \eg computing $Min_\prec$ of a set of refinements. Furthermore, the runtime of LSt increases exponentially with the number of attributes, as additional iterations must be performed in order to attain the desired $\alpha$. Interestingly, for more than 6 attributes, the performance of LSt degrades and becomes slower than Basic in the worst cases. In these cases, LSt is performing close to the maximum number of allotted iterations, and incurs overhead from \eg the increasing cost to perform the rejection sampling as the estimated stable zone shrinks. However, we emphasize that this is only true in the worst case (\Cref{fig:runtime-vs-num-attrs} shows \emph{max} runtimes). On average, LSt is $11.8\times$ faster for 6 to 10 attributes.

\paragraph{\textbf{Ablation Study}}
We compare the runtime of different variants of \algoName: (LSt) \algoName\ equipped with all the optimizations, (LSt-I), denoting \algoName\ where the iteration optimization for bounded $\alpha$ is disabled, (LSt-C), denoting \algoName\ without the optimization for reducing the set of reasonable changes, (LSt-R), denoting \algoName\ with no reduction of the re-ranking cost, and (Basic) denoting \algoName with no optimizations enabled, as described in \Cref{sec:computing}.
To demonstrate the difference in preference, we report the overall running time for computing the stability values for every tuple in the top-$10$ for all values of $0\leq k\leq 10$ for each variant of the algorithm. The results are shown in~\Cref{fig:ablation_study}.

For the NBA dataset, the basic version was overall $28.6\times$ slower than LSt and $13.2\times$ slower for CSRankings. Interestingly, the performance gap between Basic and LSt-I is negligible on CSRankings, whereas on the NBA dataset Basic is $1.51\times$ slower. This is because Basic and LSt-I draw the same number of samples, but LSt-I additionally incorporates an optimization that reduces re-ranking costs, which is more effective for the NBA dataset than CSRankings due to the higher cost of the learned ranking function used for the former.
The use of iterative computation to achieve a bounded $\alpha$ value was shown to be the most effective, improving the runtime for a single tuple by a factor of $37\times$ in the best case, and $18.9\times$ overall for the NBA dataset, and $13.2\times$ for CSRankings.

A notably smaller effect was observed for the reduction of reasonable changes optimization: no change in the NBA dataset, and $1.23\times$ faster for CSRankings. The higher effect on CSRankings is because small refinements to a \emph{single} attribute can render a tuple locally unstable in this dataset, whereas in the NBA dataset, relatively complex interactions between attributes are learned by the LtR model, and so typically refinements to multiple attributes are necessary to move the tuple in the ranking. Additionally, we note that the ranking function in CSRankings is monotone, allowing for more efficient computation of the reduced~\resChanges. 
 
Finally, we see a higher gain for the optimization of reducing the cost of re-ranking in the NBA dataset, $1.51\times$ in total runtime compared to no change in CSRankings. This is due to the fact that the ranking function's runtime depends on the data size. For CSRankings, the optimization makes almost no difference, since the runtime of evaluating a simple scoring function for a small dataset is negligible when compared to the other phases of \algoName. 

\begin{figure}[t]
    \begin{minipage}{0.48\linewidth}
        \centering
        \includegraphics[clip,trim={0.25cm 0.25cm 0.25cm 0.2cm},width=0.76\linewidth]{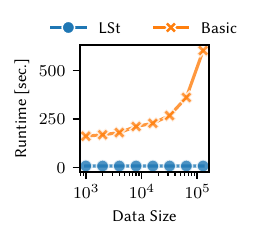}
        \caption{Data Size}
        \label{fig:runtime-vs-datasize}
    \end{minipage}\hfill%
    \begin{minipage}{0.48\linewidth}
        \centering
        \includegraphics[clip,trim={0.25cm 0.25cm 0.25cm 0.2cm},width=0.78\linewidth]{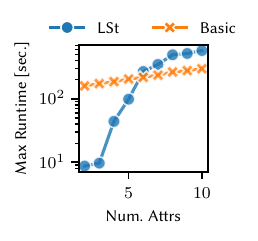}
        \caption{Num. Attrs.}
        \label{fig:runtime-vs-num-attrs}
    \end{minipage}
\end{figure}

\begin{figure}[t]
    \begin{subfigure}[t]{0.45\linewidth}
        \centering
        \includegraphics[clip,trim={0.25cm 0.25cm 0.25cm 0.2cm},width=0.85\linewidth]{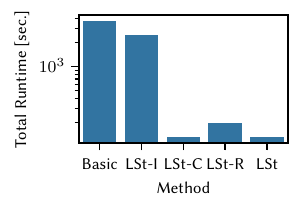}
        \caption{NBA}
        \label{fig:nba_ablation}
    \end{subfigure}
    \begin{subfigure}[t]{0.45\linewidth}
        \centering
        \includegraphics[clip,trim={0.25cm 0.25cm 0.25cm 0.25cm},width=0.76\linewidth]{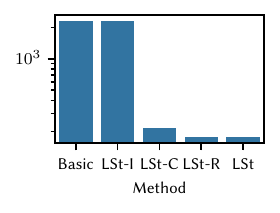}
        \caption{CSRankings}
        \label{fig:csranking_ablation}
    \end{subfigure}
    \caption{Total runtimes for experiment suites under ablation of optimizations}
    \label{fig:ablation_study}
\end{figure}

\subsection{Local vs. Global Ranking Stability}\label{sec:exp_local_vs_global}
We next demonstrate the difference between our proposed definition for local stability and the (global) stability of ranking as was defined in~\cite{AJMS18}.
We implement the Monte Carlo algorithm described in the paper for 2 dimensions, and sample 500,000 function weights to determine the global stability of the Synthetic ranking subject to changes of $a, b$ in the scoring function used to rank $\varphi(t) = a \cdot t.\texttt{x} + b \cdot t.\texttt{y}$ 
where $a, b \geq 0$. 
The original ranking, obtained by setting $a = 1$, $b = 1$, has a global stability of $0.03$ (out of $1$) while according to our definition, when taking into account the marked dense regions, most tuples are considerably locally stable (on average $0.45$) given the margins between dense regions and the set of reasonable changes described in \Cref{sec:exSetup}.

In contrast, consider a variant of the Synthetic dataset with 10 tuples, where each of the attributes of the tuple ranked $i$\textsuperscript{th} is greater by $1$ than the tuple ranked at the $i+1$ position in each attribute. Clearly, this ranking is considered perfectly stable, as each tuple dominates all the others ranked below it. However, according to our local stability measure, the stability of each tuple is at most $0.08$ for values of $k$ that don't cover the entire ranking.

\section{Related Work}
\label{sec:related-work}

\textbf{\emph{Ranking stability.}}
The stability of ranking has been studied in a line of work \cite{AJMS18,sensitivity-vectors,stability-multigroup-fairness,ranking-nutrition,mithraranking,fairly-evaluating-and-scoring}. 
Closest to our work is \cite{AJMS18}, which studies stability as the property of a \emph{ranking} with respect to a family of ranking functions, which we deemed \emph{global} stability. At a high level, stability measures how much the output may change in response to small changes in the input. Any concrete definition of stability, therefore, depends on how modifications are defined and on the notion of output effect. The specific choices for these components determine the semantics of the stability measure.
Technically, \cite{AJMS18} defines modifications as changes to the \emph{ranking function}, and considers \emph{any change} in the ranking order to constitute an output effect. In contrast, we study modifications to the \emph{data} and quantify the output effect as the \emph{magnitude of the change} in the position (\ie the value $k$) of a given tuple.
Conceptually, the stability notions in~\cite{AJMS18} are geared toward assessing how likely the (entire) ranking of the tuples is, and can be used to determine, \eg whether the ranking function was cherry-picked. In contrast, our definitions aim to measure how close tuples are in the ranking under the \emph{original ranking function}, which is especially important in scenarios where no scores are available, or the scores are non-linear in the input attributes. Note that by focusing on data changes rather than methodology, we can treat the ranking process as a black box, making our framework model-agnostic. Moreover, our definitions overlook changes within the dense region, thereby making our measure more useful when similar tuples are present in the ranking.

A form of local stability was studied in \cite{sensitivity-vectors}, however, the analysis in~\cite{sensitivity-vectors} relies on strong assumptions about both the ranking function and the dataset, which together enable a closed-form characterization of the perturbation magnitude required to alter the ranking. In particular, beyond restricting the class of ranking functions considered, their results apply only to a \emph{perfect season} dataset - one in which every team plays every other team exactly once, with no upsets: the top-ranked team defeats all others, the second-ranked team loses only to the first, and so forth. These assumptions render the analysis tractable but substantially limit its generality. Extending the framework of~\cite{sensitivity-vectors} to more general datasets or to alternative ranking functions is non-trivial and would require significant additional work. In contrast, our approach is model-agnostic and imposes no assumptions on the underlying data. Moreover,~\cite{sensitivity-vectors} does not explicitly provide a method to measure the local stability of a tuple when multiple attributes contribute to the score. Finally, we note that the~\cite{sensitivity-vectors} considers perturbations affecting exactly two tuples, while we examine perturbations made to a single tuple.

\textbf{\emph{Ranking explanations.}}
Beyond ranking stability, the problem of explaining ranking outcomes has been studied extensively~\cite{whynotyet,sharp,monotonic,local-explanations-global-rankings,synthesizing-scoring-functions,abductive-ranking-explanations}. Many of these works explain rankings by identifying feature importance~\cite{whynotyet,sharp,monotonic,local-explanations-global-rankings,synthesizing-scoring-functions}. While related to stability, such approaches do not explicitly quantify the margin between tuples in the ranking, as we propose.
In \cite{whynotyet,synthesizing-scoring-functions}, linear scoring functions are fit to a dataset to either match the top-$k$ of a given ranking or put a single tuple among the top-$k$. The weights of the resulting function may be interpreted as an explanation of feature importance. ShaRP \cite{sharp} adopts Shapley values to explain \emph{why} a given item ranks the way it does according to the contributions made by each of the individual attributes. Similarly, \cite{local-explanations-global-rankings} employs existing feature importance methods (\eg LIME \cite{LIME}) to explain rankings. The work of~\cite{monotonic} exploits monotonicity assumptions to derive importance measures of feature importance.
Orthogonally, \cite{abductive-ranking-explanations} finds sufficient subsets of attributes (called an abductive explanation) which, when agreed upon, result in the same ordering between tuples.

\textbf{\emph{Robustness verification in ML.}}
Verifying the robustness of machine learning models has been the subject of much recent attention \cite{probabilistic-robustness-dl-survey,proa,proven,towards-nn-robustness,esa-nn-robustness,scalable-dnn-verification,SAFARI}. Essentially, robustness quantifies how likely a perturbation to an input will change its prediction \cite{towards-nn-robustness}. Recent work has focused on \emph{probabilistic} definitions of robustness, similar to our definition of $\alpha$-local stability; we refer readers to \cite{probabilistic-robustness-dl-survey} for a recent survey. Many works (\eg \cite{proa,proven,esa-nn-robustness,scalable-dnn-verification,SAFARI}) employ a framework similar to ours, leveraging sampling and concentration inequalities to estimate robustness with guarantees on their reliability and accuracy \cite{esa-nn-robustness,proven,proa}. However, a key difference separates our definition from this line of work: the space of perturbations under consideration. Most often, robustness is considered as a function of an $\ell_p$-ball \cite{scalable-dnn-verification}. This presupposes a cost function on the perturbations enabling the determination of a bound on the maximum cost of changes for which the model is robust. However, our notion of stable zone boundaries replaces this assumed cost function with a weaker assumption: that the cost of a change is monotone non-decreasing. This enables including parts of the stable area which are not included in any $\ell_p$-ball (\eg the arms of an ``L''-shaped stable zone). Furthermore, in these works, the inference tasks are independent of other points (\ie classification and regression), while this is not the case in the ranking setting---local stability is not only a property of the ranking function, but depends on the other tuples in the ranking as well (\eg in the case of dense regions), motivating optimizations specific to this setting (\eg \Cref{sec:ranking-fewer-tuples}).

\textbf{\emph{Counterfactual explanations.}}
Counterfactual explanations provide concrete examples of how to achieve \emph{recourse} for an undesirable classification \cite{Wachter,DiCE,CERTIFAI,FACE,MACE,GeCo,FACET}.
Most of the time, these methods focus on providing a few actionable options for recourse based on minimizing the cost of changing the input \cite{Wachter,GeCo}. Instead, we focus on trying to characterize a broader set of options for recourse, essentially finding minimal recourse without specifying a cost function up front. Furthermore, many methods such as~\cite{DiCE,FACET} are tailored to specific models assuming white-box access, while our sampling-based approach allows us to be model-agnostic.

\textbf{\emph{Local model-agnostic explanations.}}
Local model-agnostic explanations \cite{LIME,SHAP,sharp,Anchors} answer \emph{why} a certain decision was made for the given input for any given model by treating it as a black box. In many cases, they take the form of \emph{feature importance} measures, as in \cite{LIME,SHAP,sharp}. While feature importance correlates with the local stability of a single attribute, we principally consider changes to combinations of attributes instead.

\section{Conclusion}
\label{sec:conclusion}

We recognized the need for a \emph{local view} of stability in rankings, due to the existence of dense regions in which small amounts of instability are reasonably expected and therefore tolerated. We defined this notion formally as a ratio between the refinements whose magnitude cannot incur a large change in positions to a user-defined set of reasonable changes chosen with the help of domain-specific knowledge. We showed that computing this set of changes exactly is intractable in the general case, leading us to propose a relaxed definition of local stability. We proposed \algoName, a sampling-based estimation algorithm and showed a probably-approximately-correct type guarantee by applying a concentration inequality. We further proposed \denseRegionAlgo to suggest a range of positions around a given item in the ranking with which we may reasonably expect small changes to move the item due to being ranked next to similar items. Finally, we demonstrated that our definition of local stability can provide interesting insights through case studies on real-life datasets, demonstrated the suitability of \denseRegionAlgo for its intended purpose, and provided an experimental analysis of the parameters affecting the performance of our algorithms and optimizations. 

Our work introduces many avenues for further study into the local stability of rankings. Future work may consider additional classes of refinements, such as adding or removing data, or requiring refinements to satisfy constraints (\eg denial constraints or tuples with correlated values), and handling categorical data. Extending our framework to categorical attributes may be possible in certain cases, but doing so is non-trivial, and modeling refinement magnitude appropriately depends on whether the categorical domain admits a natural order (\eg small, medium, and large sizes).

% ...

\bibliographystyle{ACM-Reference-Format}
\balance
\bibliography{main}

\appendix
\clearpage

\section{Hardness of Computing $\kStableZoneBound$}
\label{sec:apdx-hardness}

\begin{figure}
    \centering
    \hspace{-0.5cm}
    \includegraphics[width=0.4\linewidth]{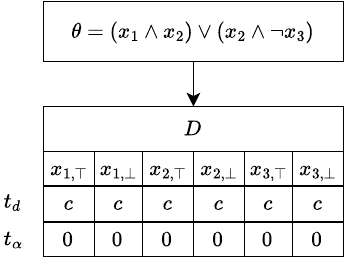}
    \caption{Example of constructed database $D$ ($c$ > 1) built from DNF formula $\theta$ for reduction in proof of \Cref{thm:counting-hardness}}
    \label{fig:reduction-construction}
\end{figure}

We begin by proving
\begin{lemma}
    \label{thm:counting-hardness}
    Given a database $D$, a tuple $t \in D$, a ranking function $f$, and a value $k$, computing $|\kStableZoneBound|$ is \textsf{\#P-hard}.
\end{lemma}
\begin{proof}
    We prove this by a parsimonious reduction from \textsf{\#DNF}. Briefly, \textsf{\#DNF} counts the number of satisfying assignments to a Boolean formula in disjunctive normal form. Furthermore, \textsf{\#DNF} is a well-known \textsf{\#P-complete} problem \cite{dnf-sharp-p-complete}. Let $\theta$ be a DNF formula consisting of clauses referring to a set of variables $x_1, \dots, x_n$. Let $\alpha$ denote an assignment to these variables, in which we say $\alpha \models \theta$ if $\alpha$ is an assignment to the variables that satisfies the formula $\theta$, and $\alpha \not\models \theta$ otherwise. We then let $\alpha(x_i)$ denote the assigned value of the variable $x_i$. 

    Now, we shall set up the reduction to our problem. Let $D$ be a database consisting of two $2n$-dimensional tuples: $v$, a tuple which shall map to an assignment $\alpha$ of the variables in $\theta$; and $w$, a ``dummy'' tuple which will swap ranks with $v$ when $v$ is a satisfying assignment. We initialize $v$ to be the $2n$-dimensional zero vector, and $w$ to be the $2n$-dimensional vector with a constant $c > 1$ in every component. Intuitively, each original variable $x_i$ for $i \in [n]$ is represented by 2 attributes in $v$; each is set to $1$ when the assignment to $x_i$ is either $\top$ or $\bot$, respectively. As we shall soon see, this separation is necessary, as otherwise the containment condition required of refinements in the set of minimal $k$-unstable refinements would become problematic. 
    
    We now define a mapping $\varphi$ from a tuple $t \in D$ to either an assignment $\alpha$, or a special value $\times$ if it cannot map to a valid assignment. For a variable $x_i$ in $\theta$, and a tuple $t \in D$, let $y_{i, \top}$ and $y_{i, \bot}$ be the attributes of $t$ corresponding to a true and false assignment to $x_i$, respectively. We first treat the case of invalid mappings. In the case that for any $x_i$, the corresponding $y_{i, \top}$, $y_{i, \bot}$ values are not in $\{0, 1\}$, then $\varphi(t) = \times$. Furthermore, if any $y_{i, \top}, y_{i, \bot}$ are \emph{both} 0 or 1, then $\varphi(t) = \times$. Now, all that is left in the mapping is to handle valid assignments. Given that the prior cases did not occur, then $\varphi(t)$ is the assignment $\alpha$ such that for each variable $x_i$ in $\theta$, $\alpha(x_i) = \top$ if $y_{i, \top} = 1$, and $\alpha(x_i) = \bot$ otherwise.

    We are now ready to construct the crux of the reduction: the ranking function. We let $f(D)$ return the ranking in which all tuples $t \in D$ with all attributes values in $[-1, 1]$ for which $\varphi(t) \models \theta$ are ranked before all tuples with all attribute values set to $c$, which are then ranked before all tuples with all attribute values in $[-1, 1]$ for which $\varphi(t) \not\models \theta$. It is easy to see then that if we are refining $v$ into $v'$, then $v'$ will rank ahead of $w$ in $f(D_{v \rightarrow v'})$ if and only if $v'$ is a satisfying assignment. Now, letting $k = 0$, we have that any refinement in $k\text{-}\textsf{SB}_{f(D)}^v$ can be mapped to a satisfying assignment of $\theta$. 
    
    All that remains is to show that $k\text{-}\textsf{SB}_{f(D)}^v$ contains all satisfying assignments of $\theta$. Recall that a refinement $\varepsilon$ of $v$ is $0$-unstable if and only if there is a refined tuple $v' \in \varepsilon(v)$ such that $\varphi(v') \models \theta$. Therefore, all we need to show is that no $0$-unstable refinement is contained in any other, which would make every $0$-unstable refinement contained in $k\text{-}\textsf{SB}_{f(D)}^t$ by definition. Towards this end, let $\varepsilon$ be a $0$-unstable refinement for $v$ over $f(D)$, and assume towards a contradiction that there is a distinct refinement $\varepsilon'$ which is $0$-unstable for $v$ over $f(D)$ and $\varepsilon' \prec \varepsilon$. Since $\varepsilon' \prec \varepsilon$, one of the magnitudes of the refinements on the attributes of $\varepsilon'$ is less than the corresponding magnitude in $\varepsilon$. Since all of the components of $\varepsilon$ are either $0$ or $1$ (as a consequence of being a refinement which has a refined tuple with a mapping to a valid assignment), and refinement magnitudes are lower bounded by 0, the difference must be between a magnitude of $1$ in one of the attributes in $\varepsilon$, while the corresponding magnitude in $\varepsilon'$ is 0. Without loss of generality, assume this is the magnitude corresponding to some $y_{i, \top}$. Note then that since the magnitude of $y_{i, \top}$ is $1$ in $\varepsilon$, then $y_{i, \bot}$ must be $0$ in order to contain a refined tuple which maps to a valid assignment. However, since $\varepsilon' \prec \varepsilon$, the magnitude of $y_{i, \bot}$ in $\varepsilon'$ must be 0, but $y_{i, \top}$ must be 0 as well. This contradicts the assumption that $\varepsilon'$ is a $0$-unstable refinement, since its refinement set has tuple mapping to a satisfying assignment (or a valid assignment at all). Therefore, there cannot be a refinement $\varepsilon'$ which is contained in $\varepsilon$ that is also $0$-stable, \ie every $0$-stable refinement for $v$ over $f(D)$ is in $k\text{-}\textsf{SB}_{f(D)}^t$.

    From this, we may conclude that every assignment $\alpha$ such that $\alpha \models \theta$ has a one-to-one correspondence with a $0$-unstable refinement in $k\text{-}\textsf{SB}_{f(D)}^v$. Therefore, given an oracle for computing $|k\text{-}\textsf{SB}_{f(D)}^t|$, we are able to answer $\textsf{\#DNF}$ in polynomial time.  
    
    We note that this is hard for any value of $k$, since we can simply add as many dummy tuples as needed.
\end{proof}

We are now ready to show
\begin{theorem}[\Cref{thm:hardness_bound_comp}]
    Unless $\textsf{FP} = \textsf{\#P}$, there is no polynomial-time algorithm for computing $\kStableZoneBound$ given a database $D$, a tuple $t \in D$, a ranking function $f$, and a value $k$.
\end{theorem}
\begin{proof}
    We prove this by showing the contrapositive: if there is a polynomial-time algorithm for computing $\kStableZoneBound$, then $\textsf{FP} = \textsf{\#P}$.
    
    A polynomial-time algorithm for computing $\kStableZoneBound$ implies that computing $|\kStableZoneBound|$ can be done in polynomial time as well, since we may simply run this algorithm and then count the number of refinements in the output $\kStableZoneBound$ (whose size is polynomially-bounded by our assumption on the running time). However, by \Cref{thm:counting-hardness}, all $\textsf{\#P}$ problems are reducible in polynomial time to computing $|\kStableZoneBound|$ via \textsf{\#DNF}. Therefore, such an algorithm implies that $\textsf{FP = \#P}$.
\end{proof}

\section{Binary Search for Reducing the Set of Reasonable Changes}
\label{sec:apdx-binary-search}
A typical feature of ranking functions is \emph{monotonicity}, \ie an improvement in one of the attributes of a tuple should only be able to improve its ranking. We can often make such an assumption since monotone ranking functions are commonly used, as they encode the expectation that the attributes are proxies for utility, and therefore by improving (raising) them, the tuple's rank should only improve \cite{monotonic}.

Recall that to reduce \resChanges, we use a portion of the sampling budget to sample single-dimensional refinements $\varepsilon_i$ for every attribute $a_i$, and let $\varepsilon^*_i$ denote the minimal (absolute) value of $\varepsilon_i$ in the sample such that $\varepsilon(t)$ is $\unstable{k}{f}{D}$. In the case that the given ranking function is monotone, we are able to find a tight lower bound (up to some constant) on $\varepsilon_i^*$ for an attribute $a_i$ in logarithmic time. Formally, a ranking function $f$ is \textbf{monotone} if and only if for every $t,~t'$ such that there is some attribute $a$ in which $t'$ is better than $t$, \ie $t'.a \geq t.a$, and equal in all other attributes, then $f(D)[t'] \leq f(D)[t]$ for every database $D$ containing $t,~t'$ (recall that we assume lower ranks are better) \cite{monotonic}. As a consequence of this definition, for a monotonic ranking function $f$ and a refinement $\varepsilon$, deciding whether every refinement $\varepsilon' \preceq \varepsilon$ is $k$-stable for $t$ over $f(D)$ becomes simple. 

Towards this end, for a refinement $\varepsilon$, let $\varepsilon^+$ and $\varepsilon^-$ be the refinements with the same magnitude as $\varepsilon$, but with all positive or negative components, respectively. If $\Delta_{f(D)}(t, \varepsilon^+(t)) \leq k$ and $\Delta_{f(D)}(t, \varepsilon^-(t)) \leq k$, then every $\varepsilon' \preceq \varepsilon$ is $k$-stable for $t$ over $f(D)$ due to the monotonicity of $f$. Finding a tight lower bound on $\varepsilon_i^*$ for an attribute $a_i$ is then a matter of running a binary search on the possible magnitudes of refinements on the attributes $a_i$ (which we obtain by discretizing the range between $0$ and $\resChanges_i$) in order to find the boundary between the $k$-stable and $k$-unstable refinements made solely by refining $a_i$.

\section{Additional Experiments}
\label{sec:apdx-addtl-expts}

\begin{figure}[h]
    \begin{subfigure}[t]{0.48\textwidth}
        \centering
        \includegraphics[width=0.96\linewidth]{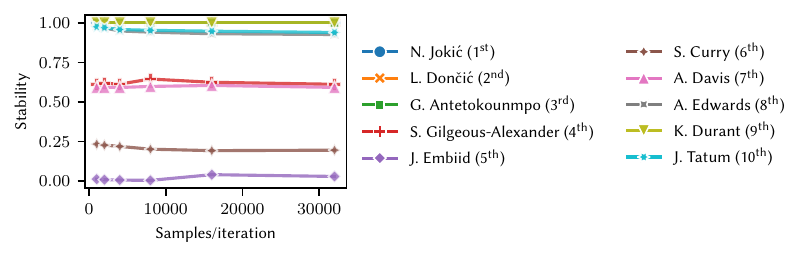}
        \caption{NBA}
    \end{subfigure}
    \begin{subfigure}[t]{0.48\textwidth}
        \includegraphics[width=0.92\linewidth]{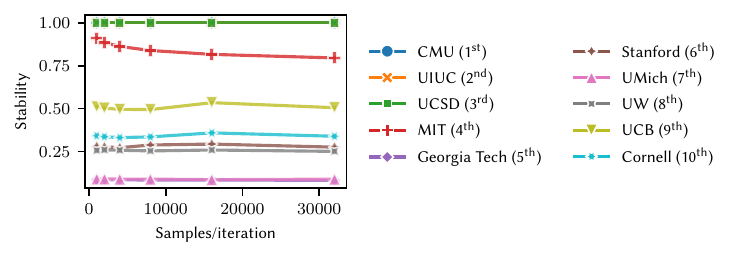}
        \caption{CSRankings}
    \end{subfigure}
    \caption{Effect of the number of samples per round taken by construction step on estimated stability}
    \label{fig:stability-vs-sample-budget}
\end{figure}

\begin{figure}[h]
    \begin{subfigure}[t]{0.48\textwidth}
        \centering
        \includegraphics[width=0.96\linewidth]{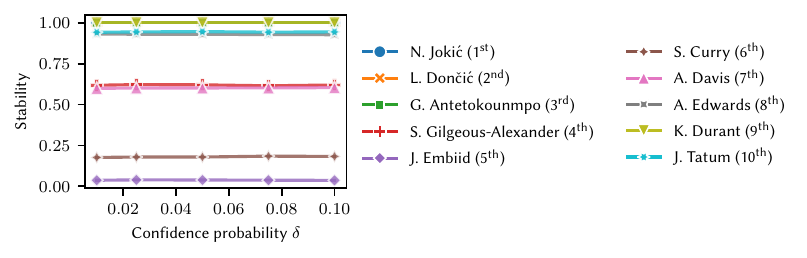}
        \caption{NBA}
    \end{subfigure}
    \begin{subfigure}[t]{0.48\textwidth}
        \includegraphics[width=0.92\linewidth]{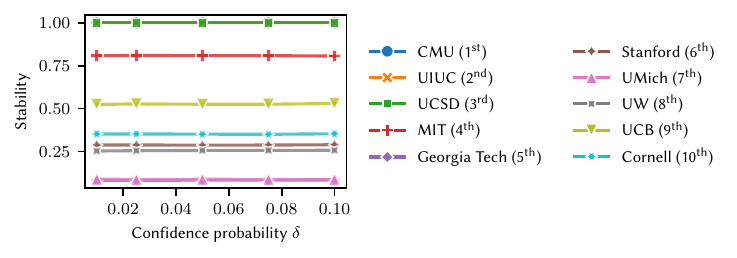}
        \caption{CSRankings}
    \end{subfigure}
    \caption{Effect of the confidence probability $\delta$ on estimated stability}
    \label{fig:stability-vs-confidence-probability}
\end{figure}

\begin{figure}[h]
    \begin{subfigure}[t]{0.48\textwidth}
        \centering
        \includegraphics[width=0.96\linewidth]{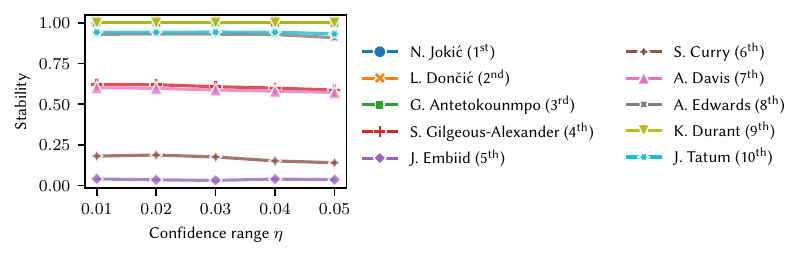}
        \caption{NBA}
    \end{subfigure}
    \begin{subfigure}[t]{0.48\textwidth}
        \includegraphics[width=0.92\linewidth]{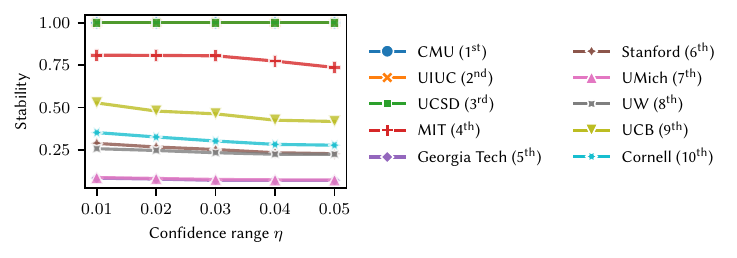}
        \caption{CSRankings}
    \end{subfigure}
    \caption{Effect of the confidence range $\eta$ on estimated stability}
    \label{fig:stability-vs-confidence-range}
\end{figure}

\begin{figure}[h]
    \begin{subfigure}[t]{0.48\textwidth}
        \centering
        \includegraphics[width=0.96\linewidth]{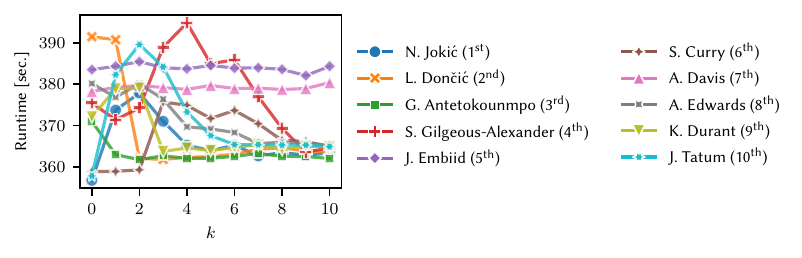}
        \caption{NBA}
    \end{subfigure}
    \begin{subfigure}[t]{0.48\textwidth}
        \includegraphics[width=0.92\linewidth]{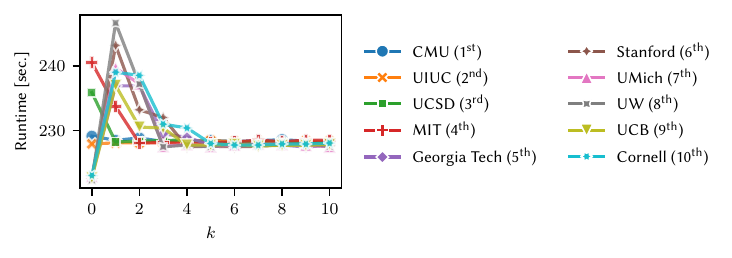}
        \caption{CSRankings}
    \end{subfigure}
    \caption{Runtime of basic \algoName for varying $k$}
    \label{fig:varying_k_noopts}
\end{figure}

\Cref{fig:stability-vs-sample-budget,fig:stability-vs-confidence-probability,fig:stability-vs-confidence-range} show that the parameters of the construction and verification steps of \algoName do not drastically alter its estimation of local stability. In some cases, the stability may go down slightly, \eg the Massachusetts Institute of Technology (MIT) in \Cref{fig:runtime-vs-sample-budget} or the University of California at Berkeley (UCB) in \Cref{fig:stability-vs-confidence-range}. These occur when additional samples are taken (either by design or by not yet having $\alpha$ beneath the desired threshold), and so the estimation of the stable zone boundary is further refined.

\Cref{fig:varying_k_noopts} shows that the basic method has a similar dependence on $k$ as the optimized version. Namely, for larger values of $k$, the runtime decreases due to the verification step taking less time when the estimated stable zone is larger (this causes the rejection sampling to take less time).

\end{document}